\documentclass[11pt, a4paper, copyright, goog, gr]{google}
\usepackage[authoryear, sort&compress, round]{natbib}
\usepackage{microtype}
\usepackage{graphicx}
\usepackage{subcaption}
\usepackage{booktabs} %
\usepackage{xcolor}

\usepackage{hyperref}
\usepackage{graphicx}
\usepackage{subcaption}
\usepackage{enumitem}
\usepackage{marvosym}

\usepackage{amsmath}
\usepackage{amssymb}
\usepackage{mathtools}
\usepackage{amsthm}

\usepackage[capitalize,noabbrev]{cleveref}

\newtheorem{proposition}{Proposition}

\theoremstyle{plain}
\newtheorem{theorem}{Theorem}[section]

\theoremstyle{definition}
\newtheorem{definition}[theorem]{Definition}

\theoremstyle{remark}

\colorlet{iclrblue}{blue!70!black}   %
\colorlet{iclrred}{BrickRed}         %

\usepackage[textsize=tiny]{todonotes}
\usepackage{multirow}
\usepackage{arydshln} %


\uselogo{} 

\title{One Model, Two Markets: Bid-Aware Generative Recommendation}
\usepackage{amsfonts,amsmath,amssymb,amsthm}

\correspondingauthor{\Letter  \ Correspondence to: Yanchen Jiang <yanchen\_jiang@g.harvard.edu>; Zhe Feng <zhef@google.com>.}

\renewcommand{\today}{2026-03-01}

\author[2,3, \Letter]{Yanchen Jiang}
\author[1, \Letter]{Zhe Feng}
\author[1]{Christopher P. Mah}
\author[1]{Aranyak Mehta}
\author[1]{Di Wang}

\affil[1]{\thepa{}{}}
\affil[2]{Harvard University}
\affil[3]{Work done during an internship at Google Research}

\begin{abstract}
Generative Recommender Systems using semantic ids, such as TIGER \citep{tiger}, have emerged as a widely adopted competitive paradigm in sequential recommendation. However, existing architectures are designed solely for semantic retrieval and do not address concerns such as monetization via ad revenue and incorporation of bids for commercial retrieval. We propose \textbf{GEM-Rec}, a unified framework that integrates commercial relevance and monetization objectives directly into the generative sequence. We introduce control tokens to decouple the decision of whether to show an ad from which item to show. This allows the model to learn valid placement patterns directly from interaction logs, which inherently reflect past successful ad placements. Complementing this, we devise a Bid-Aware Decoding mechanism that handles real-time pricing, injecting bids directly into the inference process to steer the generation toward high-value items. We prove that this approach guarantees allocation monotonicity, ensuring that higher bids weakly increase an ad's likelihood of being shown without requiring model retraining. Experiments demonstrate that \textbf{GEM-Rec} allows platforms to dynamically optimize for semantic relevance and platform revenue.
\end{abstract}

\begin{document}
\maketitle

\section{Introduction}

Recommender Systems are undergoing a paradigm shift from discriminative ranking to Generative Information Retrieval. By representing items as hierarchical Semantic IDs, models like TIGER \citep{tiger} treat recommendation as a sequence generation task. This approach has achieved state-of-the-art results by learning to predict the deep semantic structure of user trajectories \citep{singh2024better, grid, liger}.

However, applying generative models to industrial platforms presents a challenge: Monetization. Real-world systems must dynamically co-optimize for organic semantic relevance and sponsored commercial relevance, driving utility for both the user and the platform. Standard generative models fail to address this because they optimize purely for semantic likelihood. They treat every item as an organic prediction target, ignoring the economic constraints and auction dynamics that govern sponsored content.

A major difficulty is that organic items and sponsored items operate under different objectives. Organic items are selected based on estimated user preference over a corpus without price signals.  Sponsored items are selected based on a mixed objective: they must be generally relevant, but also particularly relevant to user commercial intent, incorporating bids and potential ads auction revenue. Simply training a model on a mixed stream of data confuses these signals. Furthermore, auction bids fluctuate in real-time. These dynamics are not merely economic constraints but informative signals: advertisers often adjust bids dynamically to reflect their confidence in an item's inventory, quality and relevance to the user. A model trained solely on historical logs ignores this live feedback, locking in past valuations and preventing the system from adapting to new price opportunities without retraining.

We propose \textbf{GEM-Rec}, a framework that integrates monetization into generative recommendation. We augment the semantic vocabulary with explicit control tokens. This distinguishes the decision of whether to show an ad from which item to show. By training on successful interaction logs, the model learns a ``feasibility'' policy—identifying contexts where ads were historically both semantically relevant and economically viable.

However, learning from history is distinct from maximizing current utility. Historical logs reflect a baseline of user acceptance, but they do not encode the explicit economic value of current opportunities. To bridge this gap, we introduce a Bid-Aware Decoding mechanism that injects auction bids directly into the inference process. This allows the platform to dynamically steer the generation toward high-value items in real time and fine-tune the trade-off between revenue and relevance without retraining the model.

Our contributions are as follows:
\begin{enumerate}
    \item \textbf{Unified Generation Architecture:} We propose a novel sequence formulation that integrates explicit control tokens into the Semantic ID vocabulary. This factorizes the decision of slot allocation (predicting the display mode) from content retrieval (generating the item ID), enabling a single unified model to serve both organic and sponsored requests.
    
     \item \textbf{End-to-End Learning of Marketplace Policy:} We show that the generative model implicitly learns the latent context dependencies that govern valid ad exposures. By training on successful interaction logs, the model learns to treat the decision to display an ad as an integral part of the user's trajectory.
    
    \item \textbf{Bid-Aware Constrained Decoding:} We derive an inference-time mechanism to inject bids directly into the decoding search space. We provide theoretical guarantees that this approach satisfies Allocative Monotonicity (higher bids weakly increase exposure) and Organic Integrity (ranking order for an organic recommendation remains invariant), providing practitioners with precise control over the system's operating point.
\end{enumerate}

An overview of the architecture is provided in Figure \ref{fig:gemrec}.

\begin{figure*}[t]
  \centering
  \includegraphics[width=0.9\textwidth]{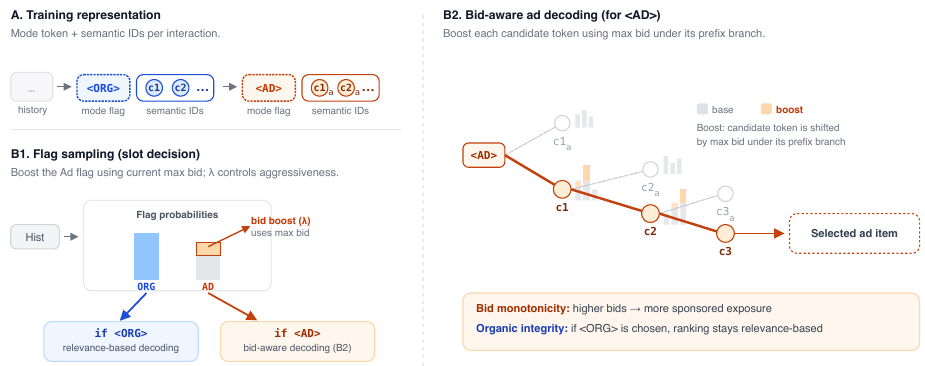}
  \caption{\textbf{GEM-Rec Unified Architecture.} Left: training on unified organic and ad logs and sampling the slot type. Right: bid-aware decoding for ads; organic decoding stays relevance-based.}
  \label{fig:gemrec}
\end{figure*}

\subsection{Related Works}

As noted earlier, the field is evolving beyond discriminative ranking over atomic IDs \citep{sasrec, bert4rec}. Early generative approaches \citep{p5} recast recommendation as conditional generation. TIGER \citep{tiger} advanced this by introducing \emph{Semantic IDs} to enable autoregressive decoding directly in a structured identifier space. This Semantic-ID paradigm has since been further explored and extended in subsequent work, including \citet{liger, grid, singh2024better, penha2025semantic, hou2025generating}. While more recent systems such as OneRec \citep{onerec} extend with scaling mechanisms and preference alignment (e.g., sparse MoE scaling and DPO-style alignment), the objectives remain fundamentally preference-centric. They lack the architectural mechanisms to process economic constraints and bid information at inference time.

In industrial feed systems, organic recommendation and sponsored delivery are often produced by separate stacks (e.g., a recommender ranking model and an ad ranking model) and then combined at serving time by a merging/blending layer \citep{chen2022hierarchically,yan2020ads,liao2022cross}. A line of research instead performs \emph{joint} optimization of recommendation and ad insertion (often via reinforcement learning), e.g., \citet{zhao2020jointly, zhao2021dear}. These RL-based insertion methods optimize whether/which/where to insert ads into a list, but they are not generative retrieval models and they do not provide a semantic-ID decoder whose outputs can be directly modulated by live bids.

A few recent papers bring generative modeling into advertising, but they study different tasks than ours. RARE \citep{rare} and EGRM \citep{grm} focus on sponsored search given a query: they generate an intermediate retrieval key for ads, whereas our setting is sequential recommendation where the model must account for how earlier organic/ad exposures affect what should be shown next. GPR \citep{gpr} studies end-to-end advertising recommendation: it represents user journeys with tokens that include both organic content and ad, but the generation target is ads (and CTR) rather than a mixed organic/sponsored recommendation list. In contrast, we generate one sequence that can include both organic and sponsored items, and we cleanly separate (i) learning when sponsored exposure is appropriate from historical logs and (ii) using the \emph{current bids} at decoding time to decide which sponsored item to show, so the system can react to real-time bid changes.

Additional related works is provided in Appendix~\ref{app:more_related}.

\section{Problem Setup}
\label{sec:problem}

We consider sequential recommendation in a marketplace containing both organic content and ads. We represent user history as a sequence of tuples $(m, i)$, where $i$ is the item and $m \in \{\text{Organic}, \text{Sponsored}\}$ is the display mode. This distinction is necessary as the same item may appear as organic in one context and as a sponsored ad in another.

We do not have access to the oracle functions governing user satisfaction or advertiser value. Instead, we learn from interaction logs, which reflect only realized successful placements. An item appears in this record only if it satisfied two latent constraints: a \emph{Platform Filter} (having a high enough bid to secure the display slot) and a \emph{User Filter} (being relevant enough to prompt a click or interaction).

Our goal is to learn a generative policy that mimics this joint acceptance distribution. By training on these logs, the model learns a ``feasibility'' policy—predicting items that historically proved to be both monetizable and relevant.

\section{The GEM-Rec Framework}
\label{sec:framework}

To solve this problem, we require a generative architecture capable of jointly modeling the decision of whether or not to present an ad, and the semantic content of an item itself. We build upon the TIGER framework \citep{tiger} but introduce a novel control structure to explicitly support this dual mandate within a single autoregressive sequence.

\subsection{Preliminaries: Semantic IDs}
Following the TIGER paradigm \citep{tiger, grid}, we represent items using Semantic IDs. This approach moves beyond atomic integers by mapping items to hierarchical tuples of discrete codes $s_i = (c_1, \dots, c_D)$ derived from a Residual Quantized VAE (RQ-VAE).

Crucially, this quantization is coarse-to-fine: the first code $c_1$ captures broad semantic categories, while subsequent codes capture increasingly fine-grained details. Consequently, items with similar content embeddings map to IDs that share common prefixes. This structure allows the model to generalize relevance patterns across the inventory and enables the prefix-based constraints we utilize in decoding.

\subsection{Unified Sequence Construction}
Recall from Section \ref{sec:problem} that the user history consists of context-aware tuples $(m, i)$. To ingest this structured history into a Transformer, we flatten it into a linear stream by treating the display mode as a prefix modifier.

We augment the vocabulary with two special control tokens: $\mathcal{F} = \{ \texttt{<ORG>}, \texttt{<AD>} \}$.
For an item at step $t$ with semantic codes $(c_{t,1}, \dots, c_{t,D})$, the generative sequence segment $\mathbf{x}_t$ is constructed as:
\begin{equation}
    \mathbf{x}_t = [f_t] \oplus [c_{t,1}, c_{t,2}, \dots, c_{t,D}]
\end{equation}
where $f_t = \texttt{<AD>}$ if $m_t = \text{Sponsored}$, and $\texttt{<ORG>}$ otherwise. The full autoregressive sequence is the concatenation of these segments: $\mathbf{x} = \mathbf{x}_1 \oplus \mathbf{x}_2 \oplus \dots \oplus \mathbf{x}_T$.

\textbf{Unified Representation and Strategic Flexibility} 
It is important to note that the Semantic ID generation process remains unchanged. We utilize standard item codes $s_i$ derived from a pre-trained RQ-VAE, ensuring the model's understanding of item content is consistent with standard generative baselines.

Our core innovation lies in the introduction of the control token $f_t$, which acts as a learnable mode switch for the generative process. Since the interaction is represented as $[f_t] \oplus s_i$, the Transformer's attention mechanism processes the item content differently depending on the prefix:

\begin{enumerate}
    \item Under \texttt{<ORG>}: The model operates in ``Preference Mode," retrieving items that maximize semantic match to the user's organic history.
    \item Under \texttt{<AD>}: The model shifts to ``Monetization Mode." It learns to target the subset of inventory that historically succeeded as sponsored impressions. By training on successful historic clicks, the model implicitly captures the intersection of high semantic relevance and high economic value, favoring items that are likely to both win an auction and satisfy the user.
\end{enumerate}

This design structurally factorizes the recommendation task. By generating the control token first, the model explicitly decides the intent of the slot (Organic vs. Sponsored) before committing to the specific content (Item ID). This modularity is crucial for the inference-time flexibility discussed in Section \ref{sec:mechanism}.

\subsection{The Factorized Generative Objective}
To operationalize this structure, we train a Transformer to minimize the negative log-likelihood of the sequence $\mathbf{x}$. Letting $H_{<t}$ denote the sequence history (context) prior to step $t$, the probability of generating the next interaction at step $t$ factorizes as follows:

\begin{equation}
    P_\theta(\mathbf{x}_t | H_{<t}) = \underbrace{P_\theta(f_t | H_{<t})}_{\text{Ad Satisfaction Modeling}} \cdot \underbrace{\prod_{k=1}^D P_\theta(c_{t,k} | H_{<t}, f_t, c_{t,<k})}_{\text{Mode-Conditional Retrieval}}
\end{equation}

This factorization effectively disentangles the platform's two distinct responsibilities:

\begin{enumerate}
    \item \textbf{Ad Satisfaction Modeling} ($P(f)$): The first term learns the contextual boundaries of monetization. By maximizing likelihood over realized trajectories, the model internalizes the latent constraints governing ad exposure (e.g., fatigue or narrative disruption). It aligns the probability of the $\texttt{<AD>}$ flag with the historical density of successful ad interactions, implicitly learning to suppress sponsored slots in contexts where they would degrade user utility.
    
    \item \textbf{Mode-Conditional Retrieval} ($P(c|f)$): The second term learns semantic relevance. Conditioned on the chosen slot type, the model retrieves the item best suited for that specific mode. If $f_t=\texttt{<ORG>}$, the model optimizes for pure organic preference. If $f_t=\texttt{<AD>}$, the model retrieves items that maximize likelihood within the distribution of historically clicked sponsored items. This effectively learns to rank items that possess both high semantic relevance and the commercial viability implicit in the training logs.
\end{enumerate}

\section{Inference-Time Bid Modulation}
\label{sec:mechanism}

The framework established in Section \ref{sec:framework} allows the model to learn a ``Safe Baseline" from historical logs. However, this approach does not fully incorporate dynamic commercial relevance and economic utility. Because the training distribution reflects only a historical baseline of ad engagement, the model remains blind to real-time market opportunities, treating low-value and high-value ads identically if their semantic relevance is similar.

To incorporate these dynamics, we introduce GEM-Decoding, a mechanism that injects active bids ($b$) into the inference step. This allows the platform to dynamically "boost" the probability of ads based on their real-time economic value, effectively steering the generative process without retraining.

\subsection{Two-Level Logit Modulation}
Consistent with our factorized architecture, we inject bid information at two levels: the Slot Decision (\textit{whether} to show an ad) and the Item Decision (\textit{which} ad to show). The strength of this modulation is controlled by a scalar $\lambda \ge 0$.

\textbf{1. Slot-Level Modulation (Dynamic Ad Load)}
First, we modulate the decision to open a sponsored slot. We wish to encourage the model to open a sponsored slot if the available inventory is valuable. Let $z_{\texttt{<AD>}}$ be the raw score (logit) for the sponsored flag. We boost this score using the \textbf{maximum bid} among valid sponsored candidates ($b_{max}$):
\begin{equation}
    \tilde{z}_{\texttt{<AD>}} = z_{\texttt{<AD>}} + \lambda \cdot \log(1 + b_{max})
\end{equation}
If high-value inventory is available ($b_{max}$ is high), this term encourages the selection of $\texttt{<AD>}$ flag more frequently. This allows Ad Rate to scale dynamically with market demand.

\textbf{2. Item-Level Modulation (Maximizing Revenue)}
Conditioned on the sampled flag being \texttt{<AD>}, the model must retrieve the most valuable item. To guide the beam search effectively, we employ Prefix-Aware Bid Aggregation.

Since Semantic IDs are hierarchical, an intermediate token $c_k$ represents a cluster of items sharing that prefix. We pre-compute a lookup table $\mathcal{B}(c_k | c_{<k})$ that stores the maximum bid of any valid item within that semantic cluster. Conditioned on the flag $f_t=\texttt{<AD>}$, we then modulate the token logits:
\begin{equation}
\tilde{z}_{c} = z_{c} + \lambda \cdot \log(1 + \mathcal{B}(c))
\end{equation}
This ensures Bid-Aware Logit Shaping: among semantically plausible tokens, the decoder is biased toward branches containing high-bidding items. This steers the generation path toward valuable outcomes early in the sequence, effectively pruning low-value branches before they are fully generated.

\subsection{Pricing Mechanism}
\label{sec:pricing}

For our experimental evaluation, we adopt the First-Price payment rule, where the winning advertiser pays their bid. This aligns with the recent structural shift in the digital advertising ecosystem, where major platforms have transitioned to first-price auctions to reduce complexity and improve transparency \citep{bigler2019updateFirstPriceAdManager, googleAdSense2021movingFirstPriceAuction, openx2017transparentFirstPriceBlog, pubmatic2018firstPriceAuctionsAuctionDynamics, microsoftxandr2025auctionOverview, amazonaps2023apsUpdateFirstPrice, magnite2021ctvMovingToFirstPrice, sivan2020competitive}. Accordingly, we report ``Revenue" in our results as the cumulative sum of the winning bids for all generated ad slots.

While First-Price auctions are increasingly common in practice and allow the system to operate directly on observable bids, ideally, we would also want to guarantee dominant-strategy incentive compatibility (DSIC) to ensure these bids truthfully reflect underlying advertiser values.
In our setting, moving from first-price to truthful (e.g., critical-bid/VCG-style) payments is conceptually natural but technically nontrivial because it requires a carefully coupled interaction between decoding randomness, monotonicity of the induced allocation rule, and counterfactual evaluation. Appendix~\ref{app:stochastic_ic} discusses these tradeoffs and outlines possible directions that do not require architectural changes; we leave a full DSIC implementation to future work.

\subsection{Theoretical Properties}

We highlight two key properties of this mechanism. The proof is given in the Appendix \ref{app:theory}.

\begin{definition}[GEM-Allocation Rule]
The allocation rule $x_i(b)$ denotes the total probability that an ad item $i$ is generated given context $H$ and bid profile $b$. Reflecting the hierarchical decoding process, this factorizes into:
\begin{equation}
    x_i(b) = \underbrace{P_\lambda(f=\texttt{<AD>} \mid H, b)}_{\text{Stochastic Slot}} \cdot \underbrace{\mathbb{I}[\mathcal{D}_K(H, b) = i]}_{\text{Deterministic Item}}
\end{equation}
where $P_\lambda$ is the modulated probability of sampling the ad flag, and $\mathcal{D}_K$ is the deterministic output of beam search with width $K$ conditioned on $f=\texttt{<AD>}$.
\end{definition}

\begin{proposition}[Allocative Monotonicity]
\label{thm:bidmonotonicity}
The GEM-Allocation rule is monotone. That is, for any context $H$ and fixed opposing bids $b_{-i}$, the exposure probability $x_i(b_i, b_{-i})$ for a sponsored item $i$ is non-decreasing in $b_i$.
\end{proposition}

This confirms \emph{Allocative Monotonicity}: increasing a bid cannot reduce the likelihood that the corresponding sponsored item is shown, so the system is rationally responsive to economic signals.

\begin{proposition}[Structural Consistency]
\label{prop:consistency}
The framework satisfies three design guarantees:
\begin{enumerate}
    \item \textbf{Safe Fallback:} When $\lambda=0$, all modulated logits equal the base model logits. Therefore the system reduces to the underlying generative recommender trained on historical logs, and the ad rate is governed solely by the learned slot model $P_\theta(\texttt{<AD>}|H)$ (i.e., the model’s learned historical baseline).
    
    \item \textbf{Organic Integrity:} The modulation of logits is strictly gated by the sponsored flag ($f=\texttt{<AD>}$). Consequently, for any fixed context $H$, the relative ranking of any two organic items $i, j$ is invariant with respect to $\lambda$.
    
    This guarantees that while $\lambda$ may alter the \textit{frequency} of organic slots, it never distorts the model's assessment of \textit{which} organic content is most relevant.
    
    \item \textbf{Generalization:} In the limit where the training corpus contains only organic interactions, $P(f=\texttt{<AD>}) \to 0$. The system collapses to standard TIGER baseline.
\end{enumerate}
\end{proposition}

\subsection{Hierarchical Decoding Strategy}\label{sec:hybrid_decoding}

Because our objective explicitly factorizes the slot decision from the item generation (Section 3.3), we naturally employ a hierarchical decoding strategy:
\begin{enumerate}
    \item \textbf{Flag Sampling:} At the first step, we compute modulated logits for $\{ \texttt{<ORG>}, \texttt{<AD>} \}$ and sample the flag immediately. This enforces a hard commitment to the slot type before content generation begins. This is necessary because standard beam search maximizes joint sequence likelihood; without this constraint, the decoder would likely prune sponsored hypotheses in favor of organic sequences that inherently possess a much higher prior likelihood from the training data.
    \item \textbf{Content Beam Search:} Conditioned on the sampled flag, we run beam search with width $K$ on the \emph{modulated} logits $\tilde z$ to generate the Semantic ID tokens.

\end{enumerate}

\section{Experimental Setup}

\subsection{Constructing the Synthetic Marketplace}
To validate GEM-Rec, we require an environment with concurrent organic and sponsored inventory, along with interaction logs generated under a joint system. Standard benchmarks are insufficient as they lack the bid logs and unified interaction history necessary to capture the tension between user intent and platform monetization. In real marketplaces, training data represents a ``survivor'' distribution of items that passed both a platform filter (high enough
bid to secure the display slot) and a user filter (relevance enough to prompt a click). To approximate this, we construct a synthetic marketplace policy ($\pi_{data}$) that generates trajectories reflecting these joint constraints.

\textbf{Rationale for Simulation.} We employ this simulation as a controlled test harness to validate our core architectural contribution: the ability of GEM-Rec to integrate monetization constraints into a generative sequence. If the model can recover the implicit rules of $\pi_{data}$ (learning when users accept ads) from the synthetic logs, it demonstrates the capacity to learn analogous constraints from real-world data where these signals are explicit.

\textbf{Data Generation Process.}
We utilize four datasets, Steam \citep{sasrec}, Amazon Beauty, Sports, and Toys \citep{he2016ups}, to provide base organic inventory and user trajectories. In each dataset, we designate a random 20\% subset as ``Sponsored" candidates and assign them log-normal bids. See Appendix \ref{app:simulation} and \ref{app:construction}.

To approximate the selection funnel typical of industrial ad platforms, we generate interaction trajectories using a Two-Stage Policy:

\begin{enumerate}

\item \textbf{Stage 1 (Retrieval): Semantic Relevance Filter.} The policy first retrieves a candidate set of sponsored items that share a sufficiently deep semantic prefix with the user's organic intent. This enforces a ``User-Centric" constraint, mirroring the reality that users rarely engage with contextually irrelevant ads.

\item \textbf{Stage 2 (Ranking): Probabilistic Auction.} Among the relevant candidates, the winning impression is selected via softmax-weighted sampling based on bid price. This mimics the aggregate outcomes of a competitive auction where higher bidders are statistically more likely to win impressions. \end{enumerate} 

This policy generates a training corpus consistent with a functioning marketplace: observed ad interactions are exclusively items that possessed both high semantic relevance and sufficient economic value. This provides the necessary signal for the generative model to learn a "safe" baseline for valid ad injection. 

While our primary simulated environments feature a conservative baseline ad density (roughly 3-6\%), we also evaluate the model on ablations with a higher baseline ad density of roughly 10-20\% (Appendix \ref{app:ablations}) to test its capacity to learn and adapt to different historical ad frequencies.

\subsection{Evaluation Protocols}
We evaluate GEM-Rec on three axes:
\begin{enumerate}
    \item \textbf{Strict Policy Fit (Total NDCG):} Measures how accurately the model reproduces the test set trajectories generated by $\pi_{data}$. A hit is recorded only if the model predicts \textit{both} the correct slot type (Ad vs. Organic) \textit{and} the correct Item ID. This validates if the model has recovered the latent marketplace constraints.
    
    \item \textbf{Organic Integrity (Conditional Organic NDCG):} We measure the ranking quality of organic items \textit{only} in instances where the model chooses to generate an organic slot. This explicitly tests the ``Preference Preservation" claim of Proposition \ref{prop:consistency}.
    
    \item \textbf{Economic Value \& Steerability:} We quantify the system's responsiveness to the control parameter $\lambda$ using three metrics:
    \begin{itemize}[leftmargin=*, nosep]
        \item \textbf{Ad Rate:} Percentage of generated slots that are ads.
        \item \textbf{Revenue:} Cumulative winning bids for ads shown.\footnote{Consistent with Section \ref{sec:pricing}, we utilize the First-Price mechanism (sum of winning bids) for the \emph{Revenue} metric to align with industrial standards.}
        \item \textbf{Ad Relevance:} NDCG of generated ads relative to the ground-truth user intent.
    \end{itemize}
\end{enumerate}

\subsection{Baselines}
We compare GEM-Rec against \textbf{TIGER} \citep{tiger}, the canonical framework for generative retrieval using semantic IDs. Since TIGER lacks a mechanism for sponsored slots, it serves as the ``Pure Utility" reference point—representing an unconstrained model optimized solely for organic relevance.

\subsection{Implementation Details}
We implement GEM-Rec in PyTorch. As the official source code for TIGER \citep{tiger} is not publicly available, we adopt the codebase from \citep{liger}. We strictly utilize their reproduction of the base TIGER architecture to ensure a fair and reproducible baseline, without employing the \citet{liger}-specific training objectives. We extend this backbone to support the unified organic-sponsored vocabulary and the inference-time GEM-Logits mechanism. We provide the model details in Appendix \ref{app:implementation} and will open-source the code after conference acceptance and publication.

\subsection{Steerability and Ad Rate Control}
A core requirement for industrial recommendation is the ability to throttle ad density dynamically. At $\lambda=0$, the model roughly reproduces the baseline ad density ($\sim$3.5\%) inherent in the training logs. An increase in Ad Rate is structurally guaranteed by the logit modulation. Figure \ref{fig:steer} empirically validates this, showing that the scaling is smooth and predictable. As $\lambda$ increases, the logit boost smoothly shifts the distribution without causing an abrupt collapse or saturation of the generative policy. This confirms that the learned control tokens successfully disentangle slot allocation from content generation.

This confirms that the learned control tokens successfully disentangle slot allocation from content generation. This control is robust: we verified that the model maintains a 100\% validity rate for generated ad identifiers across all settings, ensuring that high bid pressure does not induce hallucinations (see Appendix \ref{app:validity_analysis}).

\begin{figure*}[t!]
    \centering
    \begin{subfigure}[b]{0.48\textwidth}
        \centering
        \includegraphics[width=0.95\linewidth]{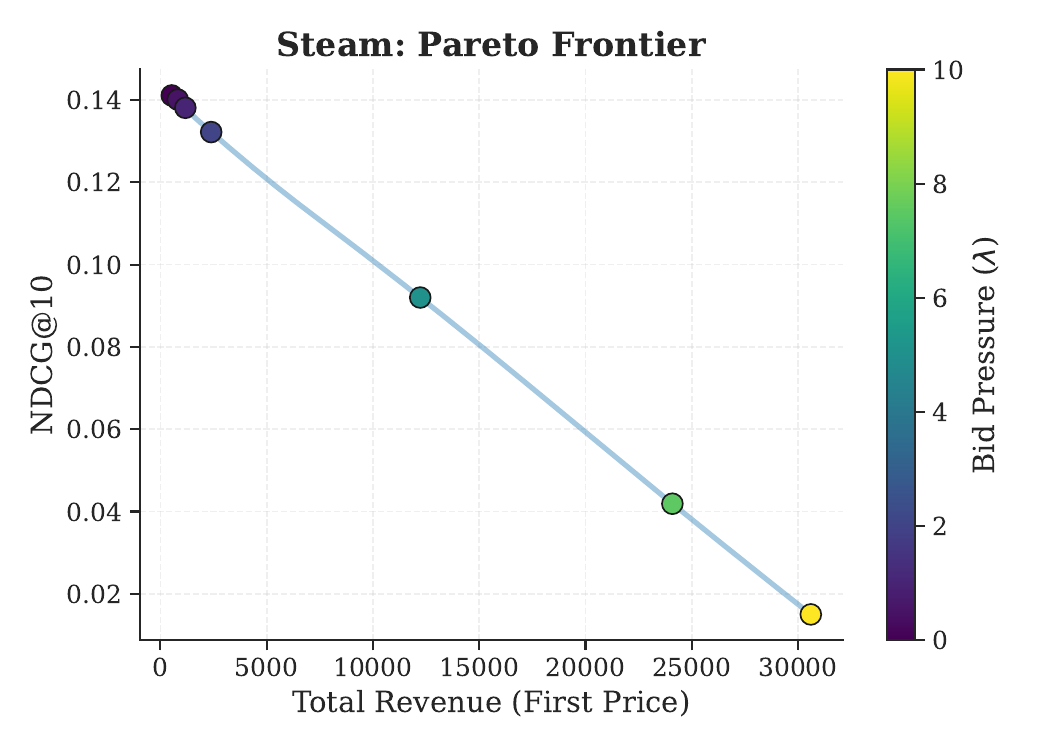}
        \caption{\textbf{Pareto Frontier:} Total Welfare vs. Strict Policy Fit.}
        \label{fig:pareto}
    \end{subfigure}
    \hfill
    \begin{subfigure}[b]{0.48\textwidth}
        \centering
        \includegraphics[width=0.95\linewidth]{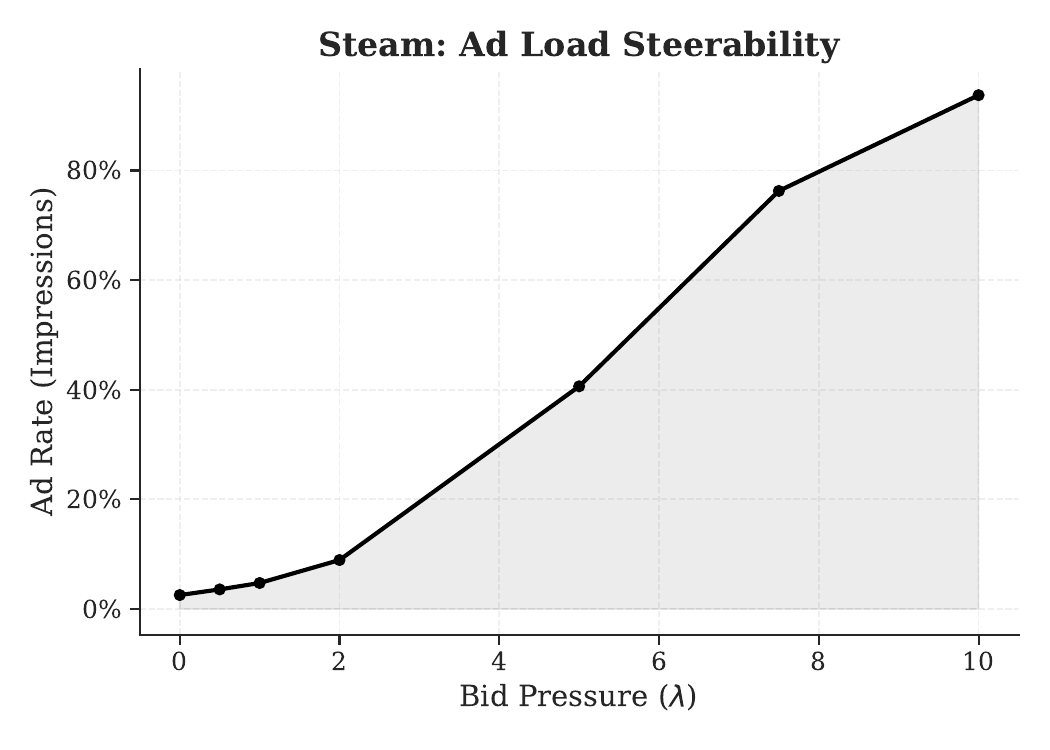}
        \caption{\textbf{Steerability:} Ad Rate Control via $\lambda$.}
        \label{fig:steer}
    \end{subfigure}
    \caption{\textbf{Macro-Dynamics of GEM-Rec. on the Steam dataset} (a) The trade-off between Platform Utility (Revenue) and Policy Fit (NDCG with respect to the training data). We observe a frontier where significant platform utility can be generated before the model deviates significantly from the marketplace policy. (b) The parameter $\lambda$ provides monotonic control over Ad Rate.}
    \label{fig:dynamics}
\end{figure*}

\subsection{The Pareto Frontier: Revenue vs. Policy Fit}
Figure \ref{fig:pareto} illustrates the trade-off between Total Revenue (Allocative Efficiency) and Total NDCG. Ideally, a system should allow the platform to increase economic utility without deviating significantly from learned user policy.

As $\lambda$ increases, revenue rises while Total NDCG declines. This decline is expected: high $\lambda$ values force the model to systematically substitute organic items with ads, creating a proportional divergence from the organic-heavy test set. On the Steam dataset (Figure \ref{fig:pareto}), we observe a near-linear frontier, providing a highly predictable way to trade off policy fit for revenue. We note that the shape of this frontier is partly governed by the underlying data distribution. When the environment contains a higher historical density of ads, for example in the ablations in Appendix \ref{app:ablations}, the model learns that sponsored content is viable across a broader range of sequential contexts. Consequently, when asked to increase the ad rate via $\lambda$, the model can initially satisfy this demand by utilizing these naturally viable slots, creating a more convex frontier with a distinct high-efficiency region before the steeper trade-off begins (e.g., see the high-ad-rate ablations on the Amazon Sports and Outdoors dataset in Appendix \ref{app:convex_eg} and the Amazon Toys and Games dataset in Appendix \ref{app:convex_eg2}). Regardless of the specific curvature, the core dynamic remains consistent: the mechanism allows the platform to smoothly increase revenue without suffering disproportionate, sudden drops in overall relevance.

\begin{table*}[t]
\centering
\setlength{\tabcolsep}{5pt} %
\renewcommand{\arraystretch}{1.1} 
\caption{\textbf{Main Results: Marketplace Performance.} Comparisons of GEM-Rec against the baseline (TIGER). We report separate columns for Economic Utility, Total Metrics, and Metrics for Conditional Generation to isolate the impact of bid modulation. Train Ad \% refers to the \emph{realized} frequency of ads in the training sequence. See full table in Appendix \ref{app:additionalresults}, Table \ref{tab:appendix_full_results}.}
\label{tab:main_results} 
\resizebox{1.0\textwidth}{!}{
\begin{tabular}{ll cc cc cc}
\toprule
& & \multicolumn{2}{c}{\textbf{Economic Utility}} & \multicolumn{2}{c}{\textbf{Total Metrics}} & \multicolumn{2}{c}{\textbf{Metrics for Organic Gen.}} \\
\cmidrule(lr){3-4} \cmidrule(lr){5-6} \cmidrule(lr){7-8} 
\textbf{Dataset} & \textbf{Operating Mode} & \textbf{Ad Rate} & \textbf{Revenue} & \textbf{NDCG@10} & \textbf{Recall@10} & \textbf{O.NDCG@10} & \textbf{O.Recall@10} \\
\midrule

\multirow{3}{*}{\shortstack[l]{\textbf{Steam} \\ \scriptsize{\textit{(Train Ad \%: 3.49)}}}} 
& TIGER (Baseline) & 0.0\% & - & 0.1442$^\dagger$ & 0.1818 & 0.1487 & 0.1875 \\
& GEM-Rec ($\lambda=0.0$) & 2.5\% & 535 & 0.1411 & 0.1782 & 0.1468 & 0.1857 \\
& GEM-Rec ($\lambda=1.0$) & 4.7\% & 1,173 & 0.1381 & 0.1742 & 0.1467 & 0.1853 \\
\midrule \hdashline[1pt/2pt] \noalign{\vskip 0.2em}

\multirow{3}{*}{\shortstack[l]{\textbf{Sports} \\ \scriptsize{\textit{(Train Ad \%: 6.85)}}}} 
& TIGER (Baseline) & 0.0\% & - & 0.0176$^\dagger$ & 0.0308 & 0.0187 & 0.0327 \\
& GEM-Rec ($\lambda=0.0$) & 4.7\% & 804 & 0.0173 & 0.0335 & 0.0187 & 0.0364 \\
& GEM-Rec ($\lambda=1.0$) & 9.0\% & 1,804 & 0.0168 & 0.0323 & 0.0185 & 0.0359 \\
\midrule \hdashline[1pt/2pt] \noalign{\vskip 0.2em}

\multirow{3}{*}{\shortstack[l]{\textbf{Toys} \\ \scriptsize{\textit{(Train Ad \%: 4.52)}}}} 
& TIGER (Baseline) & 0.0\% & - & 0.0278$^\dagger$ & 0.0536 & 0.0291 & 0.0561 \\
& GEM-Rec ($\lambda=0.0$) & 3.5\% & 319 & 0.0277 & 0.0536 & 0.0295 & 0.0570 \\
& GEM-Rec ($\lambda=1.0$) & 7.1\% & 717 & 0.0272 & 0.0522 & 0.0298 & 0.0575 \\
\midrule \hdashline[1pt/2pt] \noalign{\vskip 0.2em}

\multirow{3}{*}{\shortstack[l]{\textbf{Beauty} \\ \scriptsize{\textit{(Train Ad \%: 4.23)}}}} 
& TIGER (Baseline) & 0.0\% & - & 0.0282$^\dagger$ & 0.0529 & 0.0293 & 0.0550 \\
& GEM-Rec ($\lambda=0.0$) & 3.1\% & 345 & 0.0301 & 0.0569 & 0.0318 & 0.0602 \\
& GEM-Rec ($\lambda=1.0$) & 6.0\% & 726 & 0.0295 & 0.0559 & 0.0320 & 0.0606 \\

\bottomrule
\multicolumn{8}{l}{\footnotesize $^\dagger$ Imputed Score: Baseline strictly predicts Organic, effectively scoring 0 on all Ad slots.}
\end{tabular}
}
\end{table*}

\subsection{Verifying Organic Integrity}
\label{sec:decoupling_res}
A central claim of GEM-Rec is that the injection of ads should not degrade the model's understanding of organic user intent. Figure \ref{fig:decoupling} provides empirical validation of the \textbf{Preference Preservation} property in Proposition \ref{prop:consistency}.

The red line shows \textbf{Total NDCG} declining as $\lambda$ increases; this is expected, as high $\lambda$ forces the model to insert ads in slots where the historical policy $\pi_{data}$ would have served organic items. However, the green dashed line, \textbf{Conditional Organic NDCG}, remains effectively flat. Similarly, Table \ref{tab:main_results} shows that \textbf{Organic Recall@10} is preserved (e.g., in Steam: 0.1853 at $\lambda=1.0$ vs 0.1857 at $\lambda=0$).

This stability confirms that the Semantic ID representations are robust. Because GEM-Decoding injects the bid term $\lambda \mathcal{V}(b)$ only into sponsored branch, the logits for the organic branch remain unperturbed, as shown in Proposition \ref{prop:consistency}. This ensures that GEM-Rec allows platforms to monetize further without introducing "noise" into the organic ranking logic.

\begin{figure*}[t!]
    \centering
    \begin{subfigure}[b]{0.44\textwidth}
        \centering
        \includegraphics[width=0.95\linewidth]{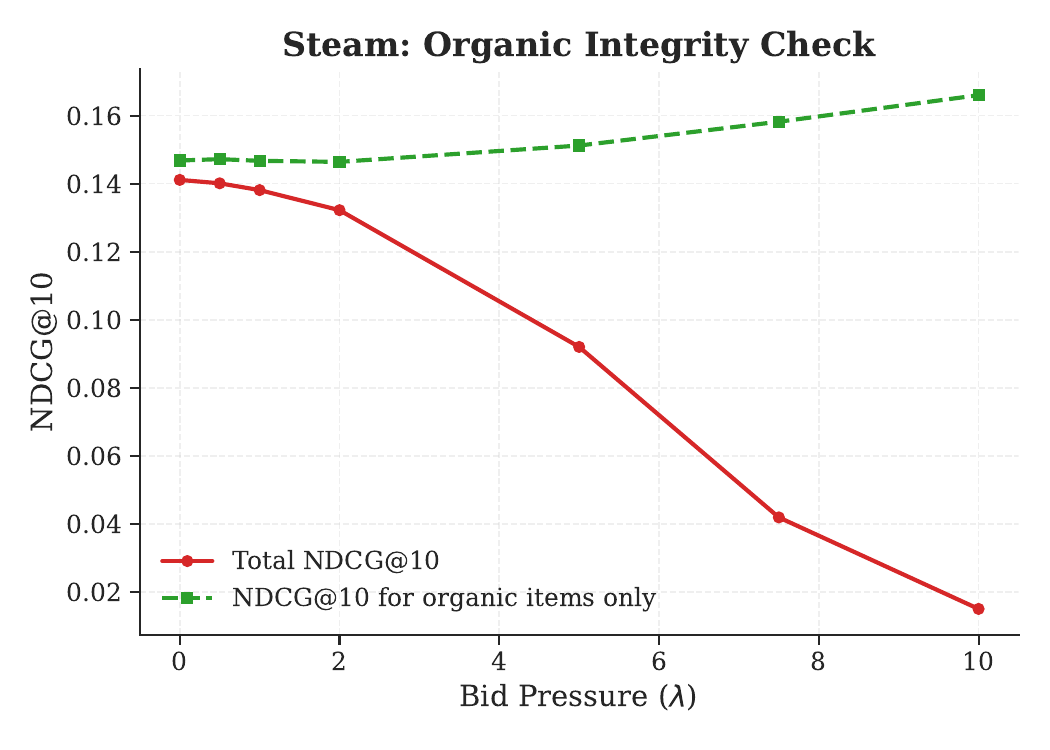}
        \caption{\textbf{Decoupling Verification:} Strict vs. Organic Integrity.}
        \label{fig:decoupling}
    \end{subfigure}
    \hfill
    \begin{subfigure}[b]{0.5\textwidth}
        \centering
        \includegraphics[width=0.95\linewidth]{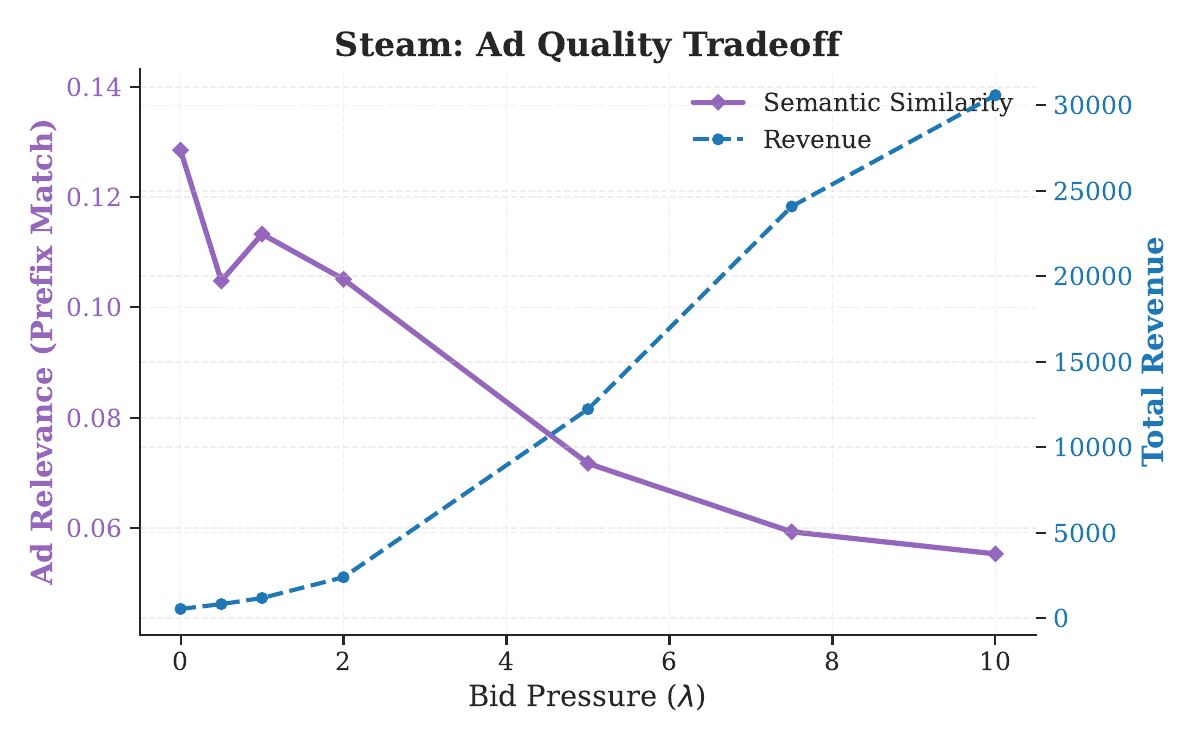}
        \caption{\textbf{Ad Quality Trade-off:} Welfare vs. Semantic Similarity.}
        \label{fig:quality}
    \end{subfigure}
    \caption{\textbf{Micro-Analysis of Recommendation Quality on the Steam dataset.} (a) A validation of our central hypothesis: while NDCG drops as we force ads (Red line), the \textit{Conditional Organic NDCG} (Green line) remains relatively flat. This shows that increasing $\lambda$ changes \textit{when} ads are shown, but does not degrade the model's understanding of user intent for organic slots. (b) The cost of monetization: higher welfare requires displaying ads with lower semantic similarity to the user's organic target.}
    \label{fig:quality_analysis}
\end{figure*}

\subsection{Ad Quality and Signal Restoration}
We utilize Prefix Match Depth to quantify the quality and relevance of generated ads. This metric measures the average number of hierarchical Semantic ID codes shared between the generated ad and the user's organic ground truth, where a depth of $\ge 2$ typically indicates sub-category alignment. 

In general, we observe that ad relevance decreases as $\lambda$ increases (Figure \ref{fig:quality}). This confirms the expected tension in the marketplace: as the system prioritizes economic value (high bids), it is inevitably forced to select items that are semantically further from the user's pure organic intent. 

This is the dominant trend across our experiments (see complete visualizations in Appendix \ref{app:additionalresults}). In general, the generative model successfully internalizes the marketplace policy during training, effectively learning to identify high-relevance items that are also valid ad candidates. Consequently, the unmodulated model ($\lambda=0$) already operates near the semantic optimum; adding further explicit bias ($\lambda > 0$) forces the model to deviate from this learned policy, degrading relevance.

We did observe minor localized exceptions (e.g., the slight non-monotonicity in Figure \ref{fig:quality} or the initial relevance gain in Appendix Figure \ref{fig:sports_full}). We attribute this to the resolution of Semantic Ambiguity. In sparse datasets, the model often correctly identifies the target \textit{category} (e.g., ``running shoes") but lacks sufficient samples to rank specific items within that cluster. However, because the training data is derived from realized auction logs, the ``ground truth" item was historically likely to be a high bidder. Consequently, injecting bid information at inference acts as a context-aware tie-breaker: it guides the model toward the high-bidder within the correct semantic cluster, effectively aligning the prediction with the historical selection bias. Despite these localized recoveries, the broader system dynamic remains a convex trade-off between short-term revenue and semantic similarity.

While our current framework successfully navigates this via a unified parameter $\lambda$, future work or practical deployments could explore decoupling this mechanism. Introducing independent scaling parameters for slot allocation ($\lambda_{slot}$) and item selection ($\lambda_{item}$) would provide platforms with even more fine-grained control, allowing them to precisely calibrate the balance between overall ad frequency and bid-aware ranking.

\subsection{Dynamic Adaptation to Market Shifts}

In real-world marketplaces, advertiser valuations are volatile; bids can spike suddenly due to seasonal events or inventory demand. A critical advantage of GEM-Rec is its inference-time plasticity—the ability to pivot towards these real-time opportunities without model retraining.

To demonstrate this, we simulate a ``Bid Shock": we randomly select 5\% of the inventory and multiply their bids by $10\times$. We evaluate the system's efficiency by measuring the \textbf{High-Value Share}: the percentage of \emph{displayed ads} that belong to this high-paying subset.

\textbf{Results}. Table \ref{tab:shock_steam_main} illustrates the system's rapid adaptation on the Steam dataset. Qualitatively similar results for all other datasets are provided in Appendix \ref{app:volatility}. When $\lambda=0$, The model relies on historical priors. While it displays some ads (2.4\%), only 21.8\% of them come from the high-value group. The system effectively ignores the shift in demand. When $\lambda=0.5$: With even modest modulation, the composition transforms. While the total Ad Rate rises moderately to 7.1\%, the High-Bid Share jumps to \textbf{81.5\%}. 

This demonstrates that GEM-Rec does not merely increase ad volume; it actively substitutes low-value impressions with high-value ones, achieving a $9\times$ revenue uplift while keeping the overall ad load conservative and controllable.

\begin{table}[h]
\centering
\caption{\textbf{Response to Market Volatility (Steam).} We analyze the composition of generated ads before and after a price shock. At $\lambda=0.5$, the system aggressively targets the high-bidders. Note that while the total Ad Rate remains low (7.1\%), the High-Value Share among ads dominates, proving the model is substituting low-value inventory for high-value inventory.}
\label{tab:shock_steam_main}
\resizebox{0.7\textwidth}{!}{
\begin{tabular}{l c c c}
\toprule
\textbf{Setting} & \textbf{Ad Rate} & \textbf{\shortstack{High-Value Share \\ (\% of Ads)}} & \textbf{Revenue} \\
\midrule
Baseline ($\lambda=0.0$) & 2.4\% & 21.8\% & $1.0\times$ \\
GEM-Rec ($\lambda=0.5$) & 7.1\% & \textbf{81.5\%} & $9.0\times$ \\
GEM-Rec ($\lambda=1.0$) & 18.0\% & \textbf{97.4\%} & $28.2\times$ \\
GEM-Rec ($\lambda=2.0$) & 62.4\% & \textbf{99.9\%} & $104.6\times$ \\
\bottomrule
\end{tabular}
}
\end{table}
\section{Conclusion}

In this work, we introduced \textbf{GEM-Rec} to bridge the structural gap between generative recommendation and marketplace monetization. By unifying organic and sponsored items within a single vocabulary, our framework internalizes complex marketplace constraints, moving beyond disjointed architectures. To ensure controllability, we developed a bid-aware decoding mechanism that disentangles semantic retrieval from economic valuation. This preserves organic user preferences while enabling inference-time adaptation to volatile bids without retraining. \textbf{GEM-Rec} establishes a robust framework for a unified generative recommender capable of simultaneously optimizing semantic relevance and platform revenue.

\bibliography{main}

\newpage
\appendix
\section*{Appendix}
\section{Additional Related Works}
\label{app:more_related}

A fast-growing literature studies how discrete item representations, tokenization, and decoding choices shape generative recommendation. \citet{grid} systematize the design space of Semantic-ID recommenders, demonstrating how indexing strategies, model architectures, and decoding constraints materially impact performance. Moving beyond fixed, precomputed indices, DIGER \citep{fu2026differentiable} introduces differentiable semantic indexing, allowing recommendation gradients to directly update the semantic index while mitigating codebook collapse. Focusing on generation dynamics, MindRec \citep{gao2025mindrec} departs from strict left-to-right autoregressive decoding, leveraging a masked diffusion process to generate recommendations in a coarse-to-fine, non-sequential manner. Finally, bridging IDs and language models, TokenRec \citep{qu2025tokenrec} proposes quantizing collaborative filtering representations into discrete tokens, optimizing how user and item semantics are tokenized and aligned for LLM-based recommendation. While these works explore alternative tokenization and decoding paradigms, our work builds upon the canonical Semantic-ID architecture \citep{tiger} to clearly isolate the impact of integrating real-time economic constraints into generative retrieval.

Our setting inherits core ideas from computational advertising and auction design.
Classical sponsored-search auctions such as the generalized second-price (GSP) mechanism are a standard reference point for bid-based ranking and pricing \citep{edelman2007internet}, but there have been an increasing ecosystem’s shift toward first-price auctions to reduce complexity and improve transparency \citep{bigler2019updateFirstPriceAdManager, googleAdSense2021movingFirstPriceAuction, openx2017transparentFirstPriceBlog, pubmatic2018firstPriceAuctionsAuctionDynamics, microsoftxandr2025auctionOverview, amazonaps2023apsUpdateFirstPrice, magnite2021ctvMovingToFirstPrice}. This shift has also motivated work studying why the change happens and analyzing equilibrium bidding and bid shading under first-price dynamics in modern ad exchanges \citep{despotakis2021first, sivan2020competitive}. In our setting, we treat bids as exogenous inputs available at serving time (as in standard ad auction interfaces), rather than modeling advertisers’ private values or strategic bid shading. We therefore use first-price payments as a transparent, low-overhead accounting rule consistent with common practice; incentive questions and DSIC-style payments are discussed separately in Appendix~\ref{app:stochastic_ic}.

On the machine learning side, a growing literature uses deep learning parameterizations to optimize auctions and other mechanisms subject to incentive constraints (e.g., learning allocation/payment networks) \citep{dutting2024optimal, ravindranath2023data, wang2024gemnet, wang2025bundleflow, feng2018deep, curry102automated}.
Our setting is different. GEM-Rec does not learn a new auction mechanism or payment network. GEM-Rec is designed around \emph{inference-time} bid modulation of a fixed generative policy: we explicitly separate (i) learning a feasibility baseline from logs from (ii) injecting current bids through constrained decoding to dynamically steer generation towards high-value items in real time, without retraining the underlying generative model or altering the standard first-price payment rule.

A parallel line of research explores the broader intersection of generative models, game theory, and mechanism design. Recent works have increasingly utilized generative AI to formalize strategic incentives, proxy human preferences in complex allocations, and align model outputs with economic constraints \citep{daskalakis2024charting, soumalias2025llm, huang2025accelerated, jiang2025incentive, shah2025learning}. While this growing literature primarily casts the generative model as an agentic participant, judge, or a behavioral proxy within economic systems, GEM-Rec shifts the focus toward platform-side market design: embedding computational advertising and real-time bids mechanics directly into the token generation and decoding processes of a Semantic ID-based generative recommender.

\section{Experiment Details}
\label{app:simulation}

\subsection{Ad Simulation Details}
To validate GEM-Rec in a controlled setting involving concurrent organic and sponsored inventory, we construct a synthetic environment derived from standard datasets. This approach ensures that the training data reflects a policy where ads are only displayed when they are both semantically relevant and contextually appropriate. We apply standard 5-core filtering to each dataset in line with \cite{liger} to define the ground-truth semantic space. 

\subsection{Model Details}
For the model architecture, we employ a \textbf{T5-based encoder-decoder}. Semantic IDs are generated via a Residual Quantized VAE (RQ-VAE) with a codebook size of 256 and a code depth of 3, trained on Sentence-T5 embeddings of the item metadata. All experiments are conducted on a single NVIDIA A100 GPU. Inference sweeps are performed with $\lambda \in [0.0, 10.0]$ to trace the full Pareto frontier.

\section{Theoretical Proofs}
\label{app:theory}

\textbf{Proof for Proposition \ref{thm:bidmonotonicity}:}

\begin{proof}
The allocation rule $x_i(b)$ factors into the slot probability $P(\texttt{<AD>})$ and the outcome of the beam search $\mathcal{D}_K$. We assume a one-to-one mapping between items and Semantic IDs, a fixed set of eligible candidates $\mathcal{I}_{elig}$ independent of bids, and that $\mathcal{D}_K$ uses deterministic beam search with fixed tie-breaking (no stochastic sampling).

We define the bid-lookup function $\mathcal{B}(c_t \mid h_{<t})$ for a next token $c_t$ given history $h_{<t}$ as the maximum bid among all items in that subtree:
$$ \mathcal{B}(c_t \mid h_{<t}) = \max \{ b_j \mid j \in \mathcal{I}_{elig} \cap \text{subtree}(h_{<t} \oplus c_t) \} $$

\textbf{Part 1: Slot Decision.}
The probability of the sponsored slot is given by $P(\texttt{<AD>}) \propto \exp(z_{\texttt{<AD>}} + \lambda \log(1 + b_{\max}))$.
Here, $b_{\max} = \max_{j \in \mathcal{I}_{elig}} b_j$. Since $\mathcal{I}_{elig}$ is fixed, $b_{\max}$ is non-decreasing in $b_i$. Since the softmax function is strictly increasing with respect to its target logit (holding others fixed), $P(\texttt{<AD>})$ is non-decreasing in $b_i$.

\textbf{Part 2: Beam Search Monotonicity.}
Conditioned on $f=\texttt{<AD>}$, the model performs deterministic beam search. An item $i$ is selected if and only if its partial hypothesis $h_i^{(t)}$ remains in the top-$K$ set at every decoding step $t=1 \dots D$, and it is the top-ranked complete sequence at step $D$.

Let $S_t(h; b)$ be the cumulative log-probability score of a partial hypothesis $h$ at step $t$ given bid profile $b$. We prove that increasing $b_i$ to $b'_i$ strictly improves (or maintains) the score of $i$'s hypothesis $h_i$ relative to any competitor $h_j$ \textbf{at every step} $t$.

The score is a sum of step-wise log-probabilities: $S_t(h) = \sum_{k=1}^t \log P(x_k \mid x_{<k})$.

\begin{enumerate}
    \item \textbf{Target Hypothesis ($h_i$):} Consider the hypothesis corresponding to item $i$. At any step $k$, the token $x_k$ is an ancestor of item $i$, so $i \in \text{subtree}(x_{<k} \oplus x_k)$.
    Since $\mathcal{B}(x_k \mid x_{<k})$ is a maximum over a set containing $i$, increasing $b_i$ implies $\mathcal{B}$ is non-decreasing. Thus, the numerator logit $\tilde{z}_{target}$ is non-decreasing.
    The softmax probability for the target token is $p_{target} = \frac{e^{\tilde{z}_{target}}}{e^{\tilde{z}_{target}} + C}$, where $C$ is the sum of exponentiated logits of all other tokens. Since $C$ is independent of $b_i$ (as $i$ only appears in the target branch), $p_{target}$ is strictly increasing in $\tilde{z}_{target}$.
    Thus, $S_t(h_i; b'_i) \ge S_t(h_i; b_i)$. The target score is \textbf{non-decreasing}.

    \item \textbf{Competitor Hypothesis ($h_j, j \neq i$):} Consider any competitor hypothesis $h_j$. Let $\tau$ be the step where $h_j$ diverges from $h_i$.
    \begin{itemize}
        \item \textbf{Shared Prefix ($k < \tau$):} The tokens are identical to $h_i$. As established above, the log-probability for these tokens is non-decreasing in $b_i$.
        \item \textbf{Divergence Step ($k = \tau$):} $h_j$ selects a token $c_{j} \neq c_{i}$. Since $j \neq i$ and they diverge here, $i \notin \text{subtree}(h_{<\tau} \oplus c_j)$. Therefore, the numerator logit for the competitor depends on $\mathcal{B}(c_j)$, which does not include $i$ in its maximization set. Thus, the competitor's numerator is constant w.r.t $b_i$.
        However, the partition function $Z$ includes the term $e^{\tilde{z}_{i}}$ for the target token $c_i$. Since $\tilde{z}_{i}$ is non-decreasing in $b_i$, $Z$ is non-decreasing.
        The log-probability is $\log p_{comp} = z_{comp} - \log Z$. Since $Z$ is non-decreasing, $\log p_{comp}$ is \textbf{non-increasing}.
        \item \textbf{Post-Divergence ($k > \tau$):} The subtree of $h_j$ is disjoint from $i$. All logits and partition functions for these steps are independent of $b_i$.
    \end{itemize}
    
    \item \textbf{Relative Rank at step $t$:} We compare the score gap $\Delta_t(b) = S_t(h_i^{(t)}; b) - S_t(h_j^{(t)}; b)$.
    \begin{itemize}
        \item For $k < \tau$, changes cancel out (identical terms).
        \item At $k = \tau$, the target term is non-decreasing, while the competitor term is non-increasing. The partition function cancels out in the difference, but the monotonic divergence remains. Thus the gap contribution increases or stays constant.
        \item For $k > \tau$, terms for $h_i$ are non-decreasing, terms for $h_j$ are constant.
    \end{itemize}
    Therefore, for any step $t$ and any competitor $h_j$, the gap $S_t(h_i^{(t)}) - S_t(h_j^{(t)})$ is monotonically non-decreasing in $b_i$.
\end{enumerate}

Since the score of $h_i^{(t)}$ cannot decrease relative to any competitor $h_j^{(t)}$, the rank of $h_i^{(t)}$ among all partial hypotheses at depth $t$ cannot worsen. 
Consequently, if $h_i$ survived the top-$K$ pruning at every step $t$ under bid $b_i$, it must also survive at every step under $b'_i > b_i$. Thus, the item $i$ remains generated.
\end{proof}

\textbf{Proof for Proposition \ref{prop:consistency}:}

\begin{proof}
These properties follow directly from the definition of the modulation mechanism in Section \ref{sec:mechanism}.
\begin{enumerate}
    \item \textbf{Safe Fallback:} By construction, the modulated logit is $\tilde{z} = z + \lambda \cdot \Delta$. If $\lambda=0$, then $\tilde{z} = z$ for all tokens, recovering the exact behavior of the base model trained on historical logs.
    \item \textbf{Organic Integrity:} The modulation term in Eq. (4) is applied only conditional on the flag $f=\texttt{<AD>}$. If the sampled flag is $f=\texttt{<ORG>}$, the decoding proceeds using unmodulated logits $z$. Thus, the relative ranking of any two organic items $i, j$ is determined solely by the pre-trained weights $\theta$, invariant to $\lambda$.
\end{enumerate}
\end{proof}

\section{Details on Dataset Construction}
\label{app:construction}
\subsection{Inventory and Bids}
For each dataset, we designate a random $20\%$ subset of the total item universe as the Sponsored Inventory $\mathcal{I}_{ad}$. To simulate a realistic long-tail distribution of advertiser valuations, we assign a static bid $b_i$ to each sponsored item $i \in \mathcal{I}_{ad}$ drawn from a Log-Normal distribution:
\begin{equation}
    b_i \sim \text{LogNormal}(\mu=0.0, \sigma=0.2)
\end{equation}
Bids are clipped at the 99.9th percentile and normalized to the range $[0.1, 1.0]$ to prevent numerical instability.

\subsection{Data Generation Policy}
We generate training trajectories by replaying the user's organic history and dynamically injecting ads. For every organic item $i_{org}$ in the user's sequence, the simulator determines if an ad should be displayed via a three-stage process:

\textbf{1. Semantic Relevance Filter.}
To ensure ads are contextually relevant, a sponsored item $j \in \mathcal{I}_{ad}$ is considered a valid candidate only if it shares a hierarchical prefix with the organic target $i_{org}$. We require a prefix match depth of $d=2$ (sharing the first two codes of the Semantic ID). If no candidates are found, the criterion is relaxed to $d=1$.

\textbf{2. Probabilistic Selection.}
Among the relevant candidates, a winner is selected probabilistically based on bid magnitude. This mimics a selection process where higher bidders are prioritized but do not have a deterministic guarantee of winning. The probability of selecting candidate $j$ is given by a softmax function over the bids:
\begin{equation}
    P(j) = \frac{\exp(b_j / \tau)}{\sum_{k} \exp(b_k / \tau)}
\end{equation}
We set the temperature $\tau = 0.1$, creating a distribution that strongly favors higher-value items.

\textbf{3. Frequency Capping.}
To model user fatigue and prevent ad saturation, we enforce a spacing constraint. The probability of accepting an ad impression is scaled by a recovery function based on the time steps $\Delta t$ since the last ad exposure:
\begin{equation}
    P(\text{Display}) = p\min(1.0, \Delta t \times r)
\end{equation}
This linear recovery scales the base acceptance rate $p$ by a fatigue factor that recovers at rate $r$, representing how people are averse against repeated exposure of ads. In the main experiments, we set the base rate $p=0.4$ and recovery rate $r=0.05$, meaning full acceptance probability is recovered after 20 steps. In the ablation study (Appendix \ref{app:ablations}), we explore a high-density setting with $p=1.0$ and $r=0.5$, which allows full recovery after just 2 steps.

\section{Truthfulness and Decoding: Tradeoffs and Constraints}
\label{app:stochastic_ic}

We adopt first-price payments in our experiments because they are simple to deploy, in line with practical trends, and align with the inference-time objective used by our decoder: once an ad is selected, the platform can charge the bid directly with no additional counterfactual computation.

At the same time, it is natural to ask whether one could instead implement a DSIC mechanism (e.g., a critical-bid or VCG-style payment rule). In mechanism design literature, DSIC is usually desired as it theoretically guarantees truthfulness for bids and prevents strategic deviations such as bid shading. Relatedly, incentive alignment has been studied in the context of generative pipelines \citep{jiang2025incentive}.

In single-shot, single-parameter auctions, however, DSIC typically relies on two ingredients: (i) an allocation rule that is monotone in each bidder's bid, and (ii) payments defined by a threshold---the minimum bid at which the bidder would still receive the allocation.

In our setting, the ``allocation rule'' is induced by an autoregressive decoding procedure. Even when the scoring function used by the decoder is monotone in bids (as in Proposition~\ref{thm:bidmonotonicity}), a full DSIC implementation must define payments with respect to the \emph{implemented} decoding algorithm. Concretely, the relevant threshold is the smallest bid at which the algorithm would still produce the same allocation under the same context and competing bids. Determining this threshold generally requires reasoning about counterfactual decoding outcomes (what would have been generated if a bidder's bid were slightly lower), which can be computationally expensive in sequential generation. 

Moreover, modern decoding pipelines often include approximations (e.g., limited-width beam search or other search heuristics). These approximations make it harder to treat the allocation as an explicit function with a tractable threshold characterization, because small bid changes can alter intermediate search states and pruning decisions. As a result, deriving and computing exact truthful payments on top of a practical decoding stack can require additional algorithmic structure beyond the model architecture itself.

Overall, first-price payments provide a robust and low-overhead baseline consistent with current practice.  Designing an exact DSIC variant for autoregressive recommendation---with provable guarantees under a realistic decoding algorithm and with acceptable compute---is an important direction for future work.

\section{Implementation Details}
\label{app:implementation}

\textbf{Model Architecture.}
We implement GEM-Rec using a T5 encoder-decoder architecture configured to match the scale of baseline implementation in \citep{liger}.
\begin{itemize}
    \item \textbf{Encoder/Decoder:} 6 layers each.
    \item \textbf{Hidden Dimensions:} $d_{model}=128$, $d_{ff}=1024$.
    \item \textbf{Attention:} 6 heads with $d_{kv}=64$.
    \item \textbf{Dropout:} 0.2.
\end{itemize}

\textbf{Semantic IDs (RQ-VAE).}
Item codes represent a hierarchy of depth 3, learned via RQ-VAE with a codebook size of 256 per level. The VAE is trained for 8,000 epochs using the AdamW optimizer.

\textbf{Training.}
The generative model is trained for 100,000 steps with a batch size of 256. We use the AdamW optimizer with a weight decay of 0.035. The learning rate is initialized at $3e^{-4}$ with a linear warmup of 10,000 steps, followed by cosine decay. All experiments are conducted on a single NVIDIA A100 GPU.

\section{Additional Experimental Results}
\label{app:additionalresults}

\subsection{Full result for main experiment}

We provide the full suite of analysis figures for the Beauty, Sports, and Toys datasets, complementing the Steam results presented in the main text. At the end, we provide Table \ref{tab:appendix_full_results} with all the detailed numerical results.

\subsubsection{Beauty Dataset}

\begin{figure*}[h!]
    \centering
    \begin{subfigure}[b]{0.48\textwidth}
        \centering
        \includegraphics[width=0.95\linewidth]{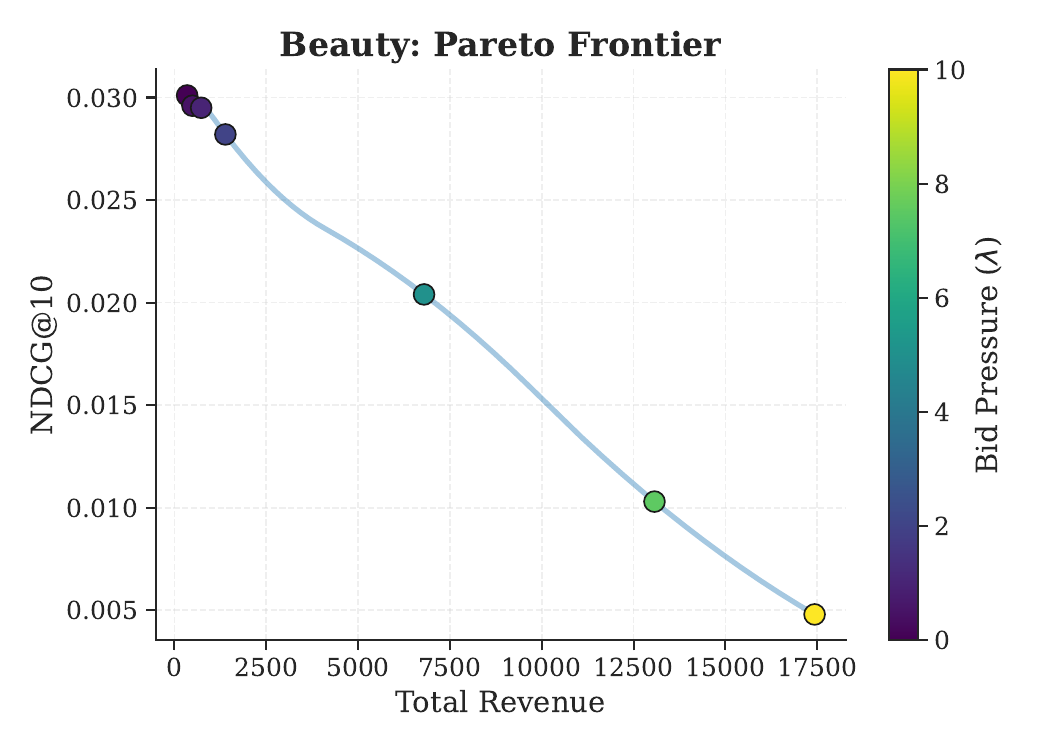}
        \caption{\textbf{Pareto Frontier} (Beauty): Revenue vs. Policy Fit.}
        \label{fig:beauty_pareto}
    \end{subfigure}
    \hfill
    \begin{subfigure}[b]{0.48\textwidth}
        \centering
        \includegraphics[width=0.95\linewidth]{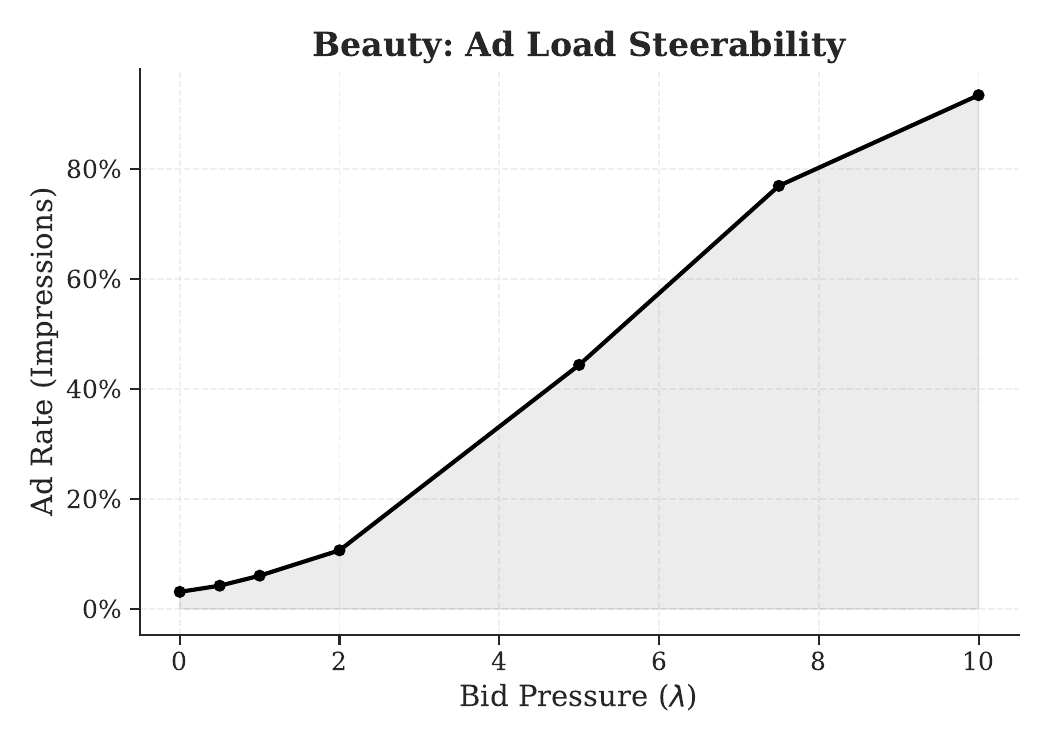}
        \caption{\textbf{Steerability} (Beauty): Ad Rate Control.}
        \label{fig:beauty_steer}
    \end{subfigure}
    
    \vspace{1em} %
    
    \begin{subfigure}[b]{0.44\textwidth}
        \centering
        \includegraphics[width=0.95\linewidth]{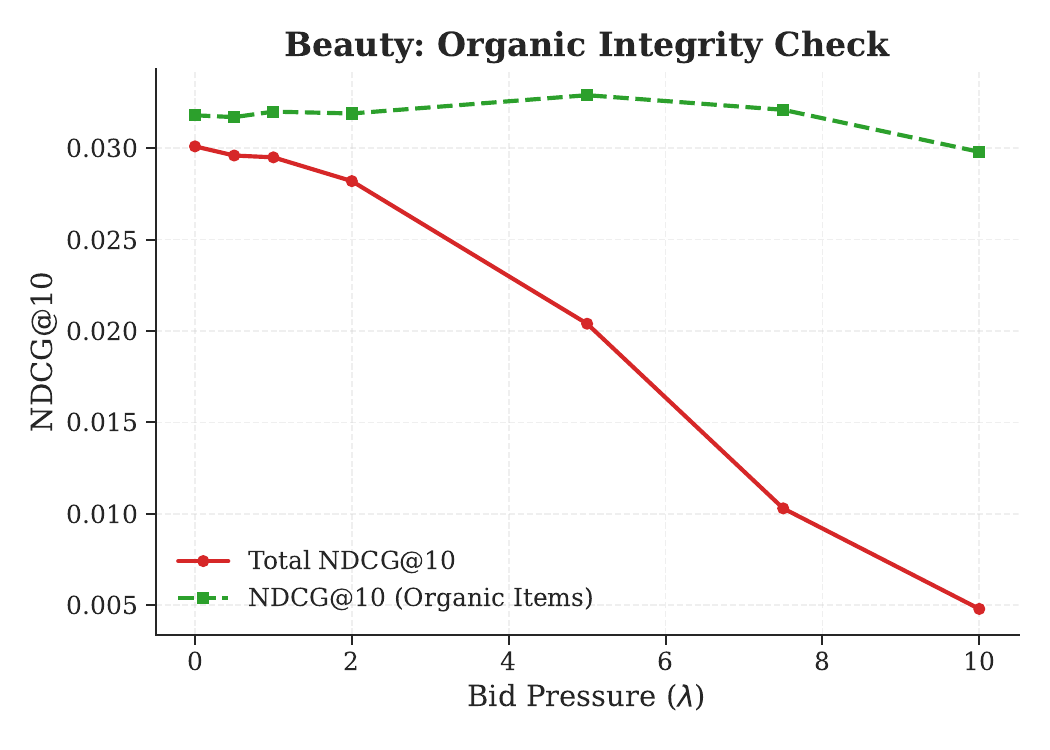}
        \caption{\textbf{Organic Integrity} (Beauty).}
        \label{fig:beauty_integrity}
    \end{subfigure}
    \hfill
    \begin{subfigure}[b]{0.5\textwidth}
        \centering
        \includegraphics[width=0.95\linewidth]{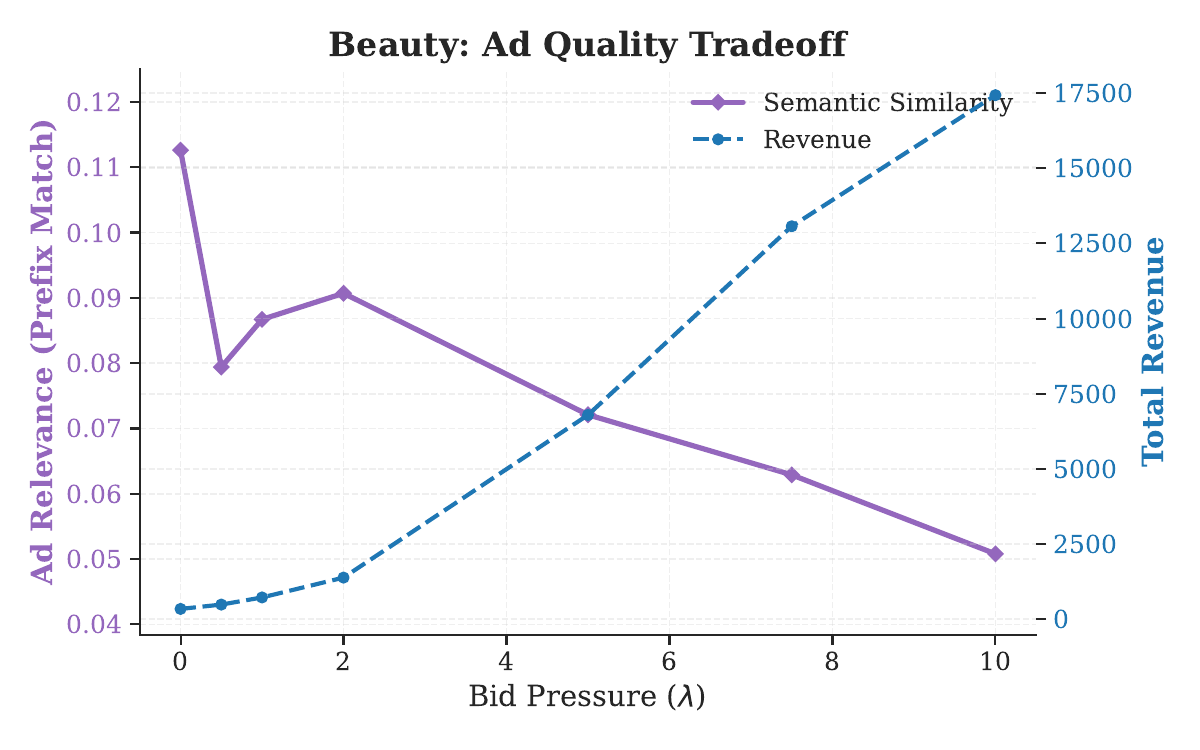}
        \caption{\textbf{Quality Trade-off} (Beauty): Welfare vs. Semantic Similarity.}
        \label{fig:beauty_quality}
    \end{subfigure}
    
    \caption{\textbf{Full Dynamics on Beauty Dataset.} Macro-dynamics (top row) and Micro-analysis (bottom row) showing consistent behavior with the Steam dataset.}
    \label{fig:beauty_full}
\end{figure*}
\clearpage

\subsubsection{Sports Dataset}
\label{app:sports_figs_main}

\begin{figure*}[h!]
    \centering
    \begin{subfigure}[b]{0.48\textwidth}
        \centering
        \includegraphics[width=0.95\linewidth]{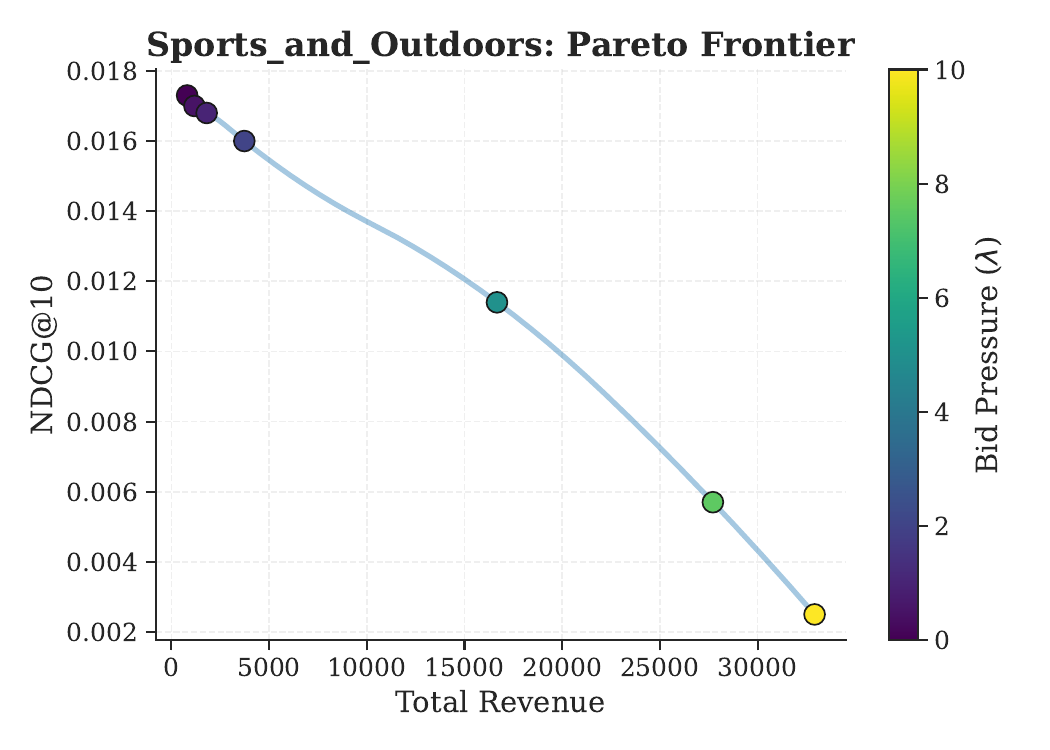}
        \caption{\textbf{Pareto Frontier} (Sports).}
    \end{subfigure}
    \hfill
    \begin{subfigure}[b]{0.48\textwidth}
        \centering
        \includegraphics[width=0.95\linewidth]{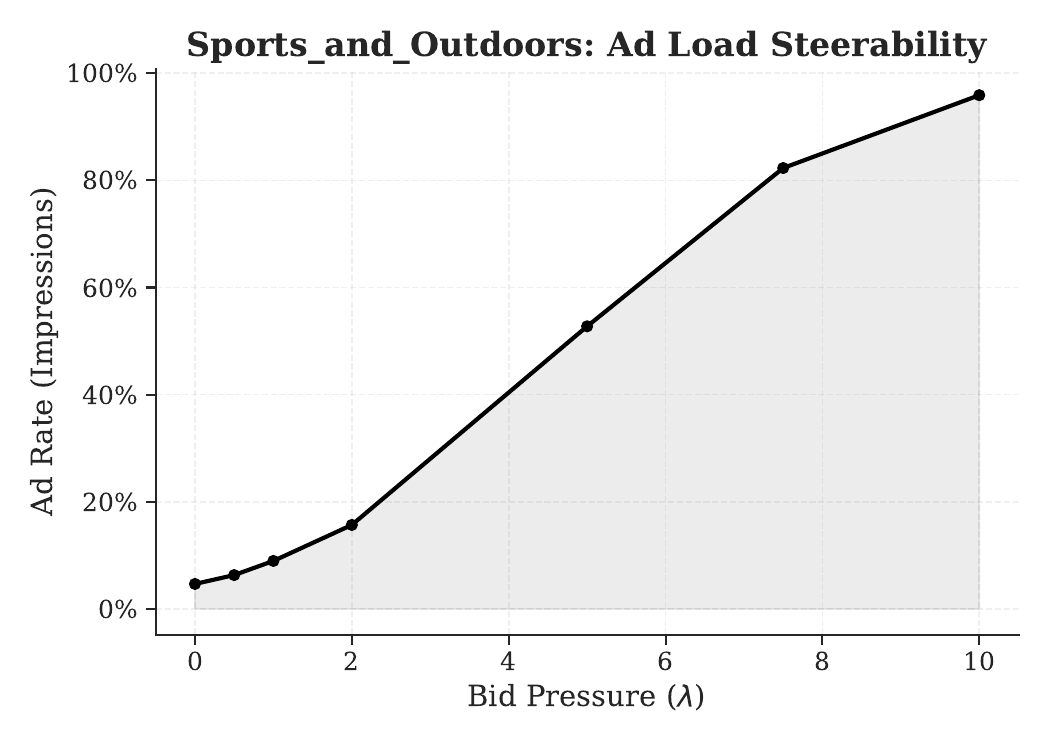}
        \caption{\textbf{Steerability} (Sports).}
    \end{subfigure}
    
    \vspace{1em}
    
    \begin{subfigure}[b]{0.44\textwidth}
        \centering
        \includegraphics[width=0.95\linewidth]{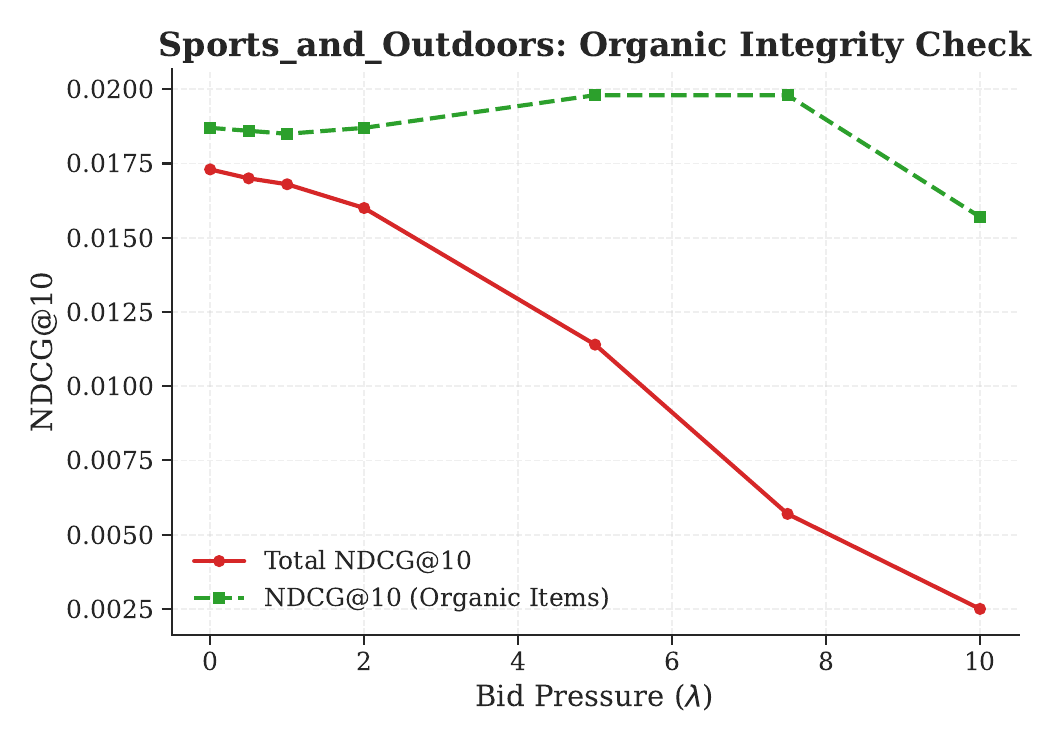}
        \caption{\textbf{Organic Integrity} (Sports).}
    \end{subfigure}
    \hfill
    \begin{subfigure}[b]{0.5\textwidth}
        \centering
        \includegraphics[width=0.95\linewidth]{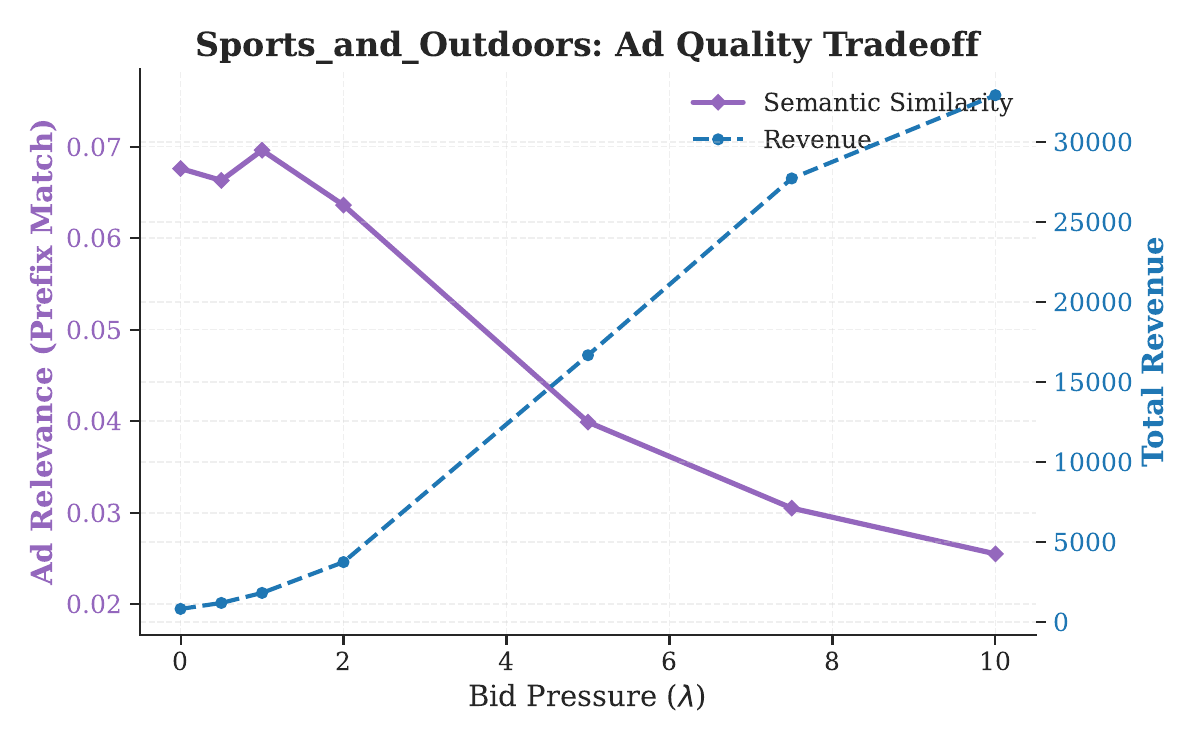}
        \caption{\textbf{Quality Trade-off} (Sports).}
    \end{subfigure}
    
    \caption{\textbf{Full Dynamics on Sports Dataset.}}
    \label{fig:sports_full}
\end{figure*}
\clearpage

\subsubsection{Toys Dataset}

\begin{figure*}[h!]
    \centering
    \begin{subfigure}[b]{0.48\textwidth}
        \centering
        \includegraphics[width=0.95\linewidth]{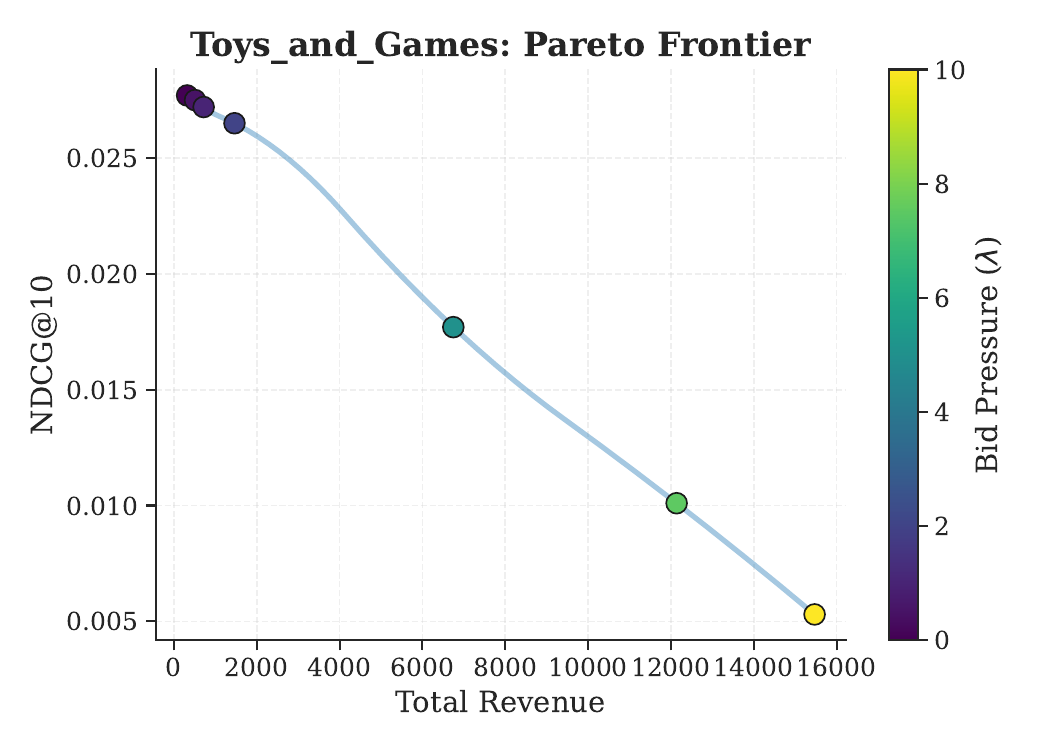}
        \caption{\textbf{Pareto Frontier} (Toys).}
    \end{subfigure}
    \hfill
    \begin{subfigure}[b]{0.48\textwidth}
        \centering
        \includegraphics[width=0.95\linewidth]{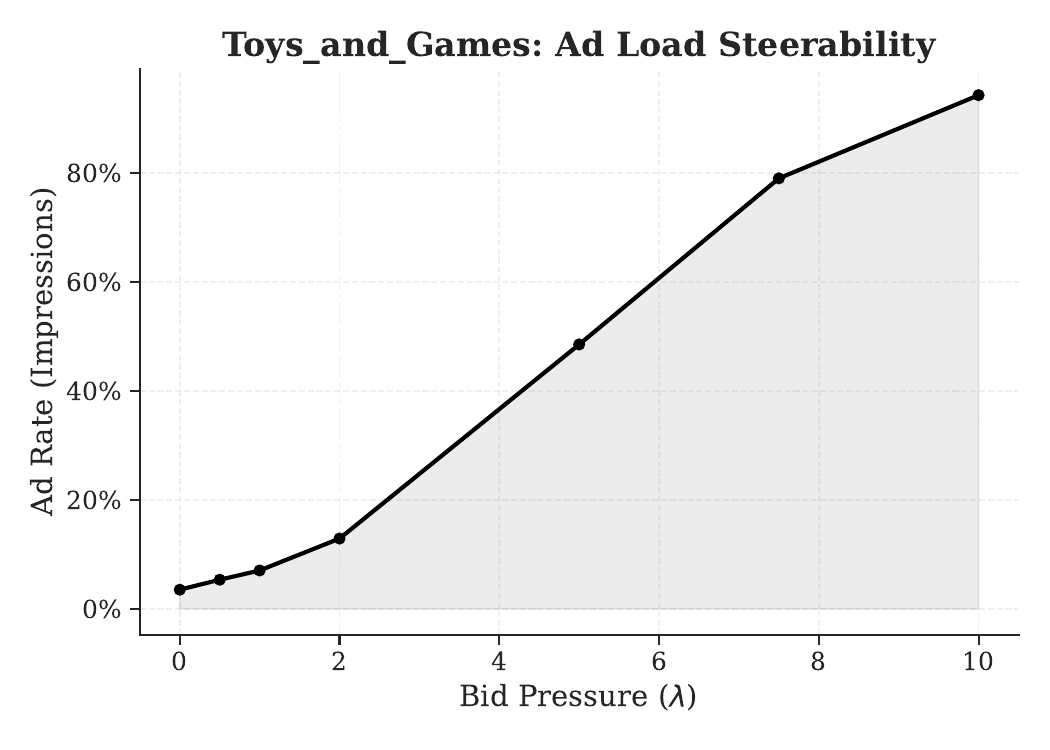}
        \caption{\textbf{Steerability} (Toys).}
    \end{subfigure}
    
    \vspace{1em}
    
    \begin{subfigure}[b]{0.44\textwidth}
        \centering
        \includegraphics[width=0.95\linewidth]{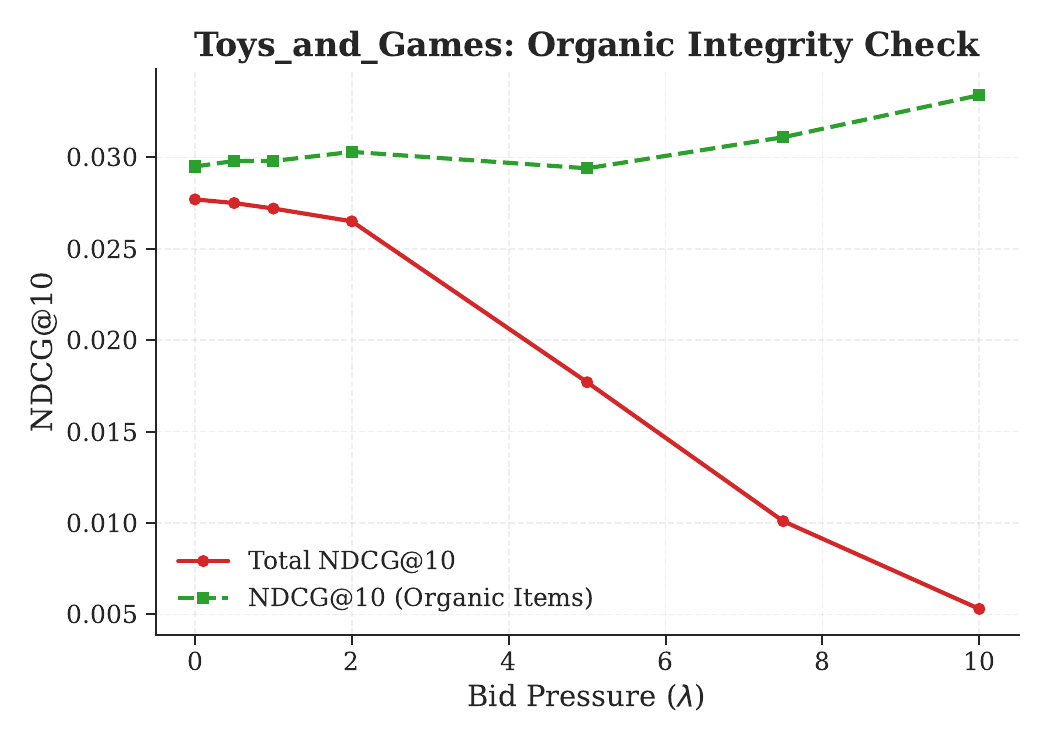}
        \caption{\textbf{Organic Integrity} (Toys).}
    \end{subfigure}
    \hfill
    \begin{subfigure}[b]{0.5\textwidth}
        \centering
        \includegraphics[width=0.95\linewidth]{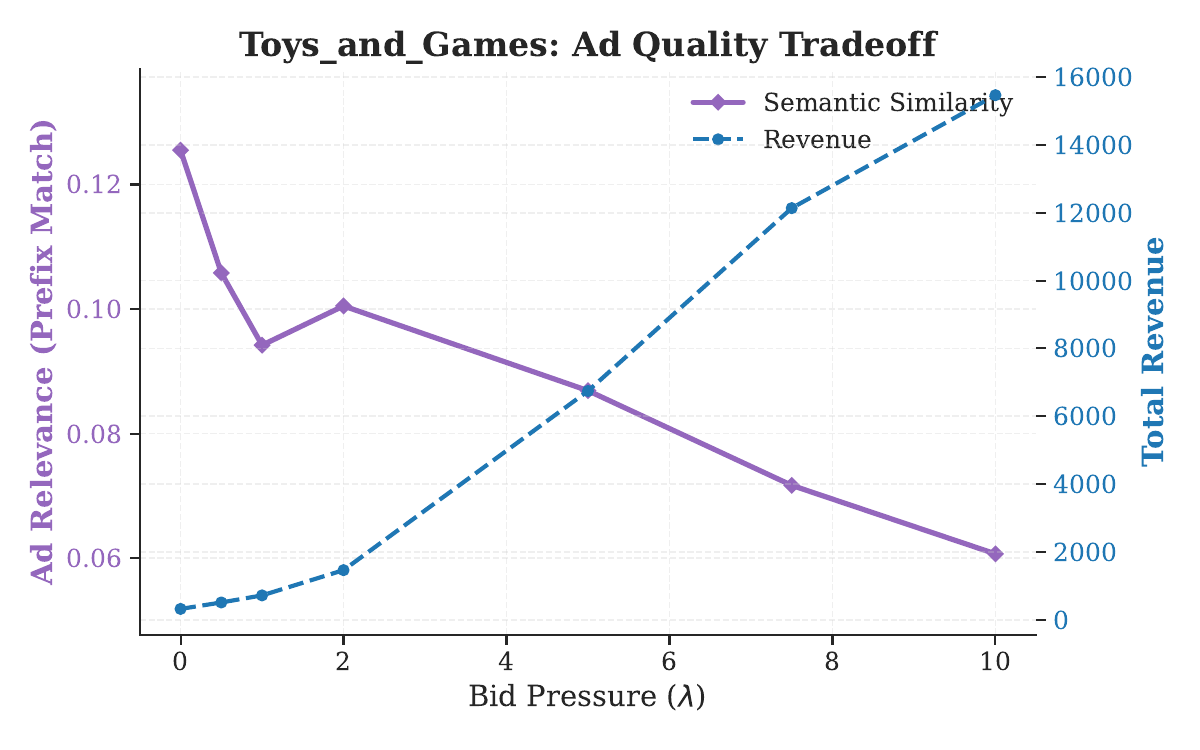}
        \caption{\textbf{Quality Trade-off} (Toys).}
    \end{subfigure}
    
    \caption{\textbf{Full Dynamics on Toys Dataset.}}
    \label{fig:toys_full}
\end{figure*}

\begin{table*}[h]
\centering
\setlength{\tabcolsep}{4pt} 
\renewcommand{\arraystretch}{1.0} 
\caption{\textbf{Full Marketplace Performance Results.} We compare GEM-Rec against the baseline (TIGER) across all datasets and $\lambda$ settings. The \textbf{Train Ad \%} column is provided for reference.}
\label{tab:appendix_full_results}
\resizebox{1.0\textwidth}{!}{
\begin{tabular}{ll c cc cc cc}
\toprule
& & \textbf{Train} & \multicolumn{2}{c}{\textbf{Economic Utility}} & \multicolumn{2}{c}{\textbf{Total Metrics}} & \multicolumn{2}{c}{\textbf{Metrics for Organic Gen.}} \\
\cmidrule(lr){4-5} \cmidrule(lr){6-7} \cmidrule(lr){8-9} 
\textbf{Dataset} & \textbf{Method} & \textbf{Ad \%} & \textbf{Ad Rate} & \textbf{Revenue} & \textbf{NDCG@10} & \textbf{Recall@10} & \textbf{O.NDCG@10} & \textbf{O.Recall@10} \\
\midrule

\multirow{8}{*}{\textbf{Beauty}} 
& TIGER (Baseline) & \multirow{8}{*}{4.23\%} & 0.0\% & 0.0 & 0.0282$^\dagger$ & 0.0529 & 0.0293 & 0.0550 \\
& GEM-Rec ($\lambda=0.0$) & & 3.1\% & 345.1 & 0.0301 & 0.0569 & 0.0318 & 0.0602 \\
& GEM-Rec ($\lambda=0.5$) & & 4.2\% & 489.3 & 0.0296 & 0.0560 & 0.0317 & 0.0603 \\
& GEM-Rec ($\lambda=1.0$) & & 6.0\% & 726.6 & 0.0295 & 0.0559 & 0.0320 & 0.0606 \\
& GEM-Rec ($\lambda=2.0$) & & 10.7\% & 1,386.8 & 0.0282 & 0.0537 & 0.0319 & 0.0610 \\
& GEM-Rec ($\lambda=5.0$) & & 44.4\% & 6,795.4 & 0.0204 & 0.0385 & 0.0329 & 0.0622 \\
& GEM-Rec ($\lambda=7.5$) & & 76.9\% & 13,068.8 & 0.0103 & 0.0201 & 0.0321 & 0.0612 \\
& GEM-Rec ($\lambda=10.0$) & & 93.4\% & 17,427.6 & 0.0048 & 0.0101 & 0.0298 & 0.0574 \\
\midrule

\multirow{8}{*}{\textbf{Sports}} 
& TIGER (Baseline) & \multirow{8}{*}{6.85\%} & 0.0\% & 0.0 & 0.0176$^\dagger$ & 0.0308 & 0.0187 & 0.0327 \\
& GEM-Rec ($\lambda=0.0$) & & 4.7\% & 804.2 & 0.0173 & 0.0335 & 0.0187 & 0.0364 \\
& GEM-Rec ($\lambda=0.5$) & & 6.3\% & 1,177.7 & 0.0170 & 0.0330 & 0.0186 & 0.0362 \\
& GEM-Rec ($\lambda=1.0$) & & 9.0\% & 1,804.2 & 0.0168 & 0.0323 & 0.0185 & 0.0359 \\
& GEM-Rec ($\lambda=2.0$) & & 15.7\% & 3,731.3 & 0.0160 & 0.0310 & 0.0187 & 0.0363 \\
& GEM-Rec ($\lambda=5.0$) & & 52.8\% & 16,662.4 & 0.0114 & 0.0223 & 0.0198 & 0.0381 \\
& GEM-Rec ($\lambda=7.5$) & & 82.3\% & 27,712.1 & 0.0057 & 0.0117 & 0.0198 & 0.0383 \\
& GEM-Rec ($\lambda=10.0$) & & 95.9\% & 32,916.6 & 0.0025 & 0.0055 & 0.0157 & 0.0330 \\
\midrule

\multirow{8}{*}{\textbf{Toys}} 
& TIGER (Baseline) & \multirow{8}{*}{4.52\%} & 0.0\% & 0.0 & 0.0278$^\dagger$ & 0.0536 & 0.0291 & 0.0561 \\
& GEM-Rec ($\lambda=0.0$) & & 3.5\% & 319.3 & 0.0277 & 0.0536 & 0.0295 & 0.0570 \\
& GEM-Rec ($\lambda=0.5$) & & 5.4\% & 509.7 & 0.0275 & 0.0530 & 0.0298 & 0.0576 \\
& GEM-Rec ($\lambda=1.0$) & & 7.1\% & 717.9 & 0.0272 & 0.0522 & 0.0298 & 0.0575 \\
& GEM-Rec ($\lambda=2.0$) & & 12.9\% & 1,462.4 & 0.0265 & 0.0502 & 0.0303 & 0.0576 \\
& GEM-Rec ($\lambda=5.0$) & & 48.6\% & 6,746.0 & 0.0177 & 0.0343 & 0.0294 & 0.0568 \\
& GEM-Rec ($\lambda=7.5$) & & 79.1\% & 12,137.0 & 0.0101 & 0.0201 & 0.0311 & 0.0600 \\
& GEM-Rec ($\lambda=10.0$) & & 94.3\% & 15,467.5 & 0.0053 & 0.0114 & 0.0334 & 0.0645 \\
\midrule

\multirow{8}{*}{\textbf{Steam}} 
& TIGER (Baseline) & \multirow{8}{*}{3.49\%} & 0.0\% & 0.0 & 0.1442$^\dagger$ & 0.1818 & 0.1487 & 0.1875 \\
& GEM-Rec ($\lambda=0.0$) & & 2.5\% & 535.3 & 0.1411 & 0.1782 & 0.1468 & 0.1857 \\
& GEM-Rec ($\lambda=0.5$) & & 3.6\% & 815.4 & 0.1401 & 0.1770 & 0.1472 & 0.1862 \\
& GEM-Rec ($\lambda=1.0$) & & 4.7\% & 1,173.6 & 0.1381 & 0.1742 & 0.1467 & 0.1853 \\
& GEM-Rec ($\lambda=2.0$) & & 8.9\% & 2,390.2 & 0.1322 & 0.1672 & 0.1464 & 0.1854 \\
& GEM-Rec ($\lambda=5.0$) & & 40.6\% & 12,224.1 & 0.0920 & 0.1173 & 0.1512 & 0.1917 \\
& GEM-Rec ($\lambda=7.5$) & & 76.3\% & 24,086.6 & 0.0419 & 0.0542 & 0.1582 & 0.1978 \\
& GEM-Rec ($\lambda=10.0$) & & 93.7\% & 30,598.5 & 0.0150 & 0.0212 & 0.1660 & 0.2070 \\
\midrule

\bottomrule
\multicolumn{9}{l}{\footnotesize $^\dagger$ Imputed Score: Baseline strictly predicts Organic, effectively scoring 0 on all Ad slots.}
\end{tabular}
}
\end{table*}
\clearpage

\subsubsection{Generative Validity}
\label{app:validity_analysis}

A critical robustness check for generative recommendation is ensuring that the model does not "hallucinate" invalid identifiers when forced to generate ads. We evaluated the validity of every generated ad sequence across all datasets and $\lambda$ settings. 

GEM-Rec achieves a 100.0\% Validity Rate for ad generation across all four datasets. This indicates that the model has learned the hierarchical structure of the Semantic IDs and output valid IDs. This observation aligns with established findings in generative retrieval, where the structured codebook facilitates robust learning of the item space. Notably, we achieve this validity purely through the model's learned weights and logit boosting, \textit{without} needing computationally expensive prefix-trie constraints during decoding.

\begin{table}[h]
\centering
\setlength{\tabcolsep}{12pt}
\renewcommand{\arraystretch}{1.1}
\caption{\textbf{Generative Validity.} The percentage of generated ad sequences that decode to a valid, existing item in the inventory.}
\label{tab:validity_simple}
\resizebox{0.3\columnwidth}{!}{
\begin{tabular}{lc}
\toprule
\textbf{Dataset} & \textbf{Valid Rate} \\
\midrule
\textbf{Beauty} & 100.0\% \\
\textbf{Sports} & 100.0\% \\
\textbf{Toys}   & 100.0\% \\
\textbf{Steam}  & 100.0\% \\
\bottomrule
\end{tabular}
}
\end{table}

\section{Volatility Experiments}
\label{app:volatility}

We provide the detailed breakdown of the ``Bid Shock" experiments across all four datasets. In each table, \textbf{High-Value Share} denotes the percentage of displayed ads that belong to the specific 5\% subset of inventory that received the $10\times$ bid multiplier.

\begin{table}[h]
\centering
\caption{\textbf{Bid Shock Response (Steam).} The baseline ($\lambda=0$) relies on historical priors and misses the new opportunity. GEM-Rec ($\lambda \ge 0.5$) dynamically pivots to capture the high-value inventory.}
\label{tab:shock_steam_app}
\resizebox{0.6\textwidth}{!}{
\begin{tabular}{l ccc}
\toprule
\textbf{Setting} & \textbf{Ad Rate} & \textbf{High-Value Share} & \textbf{Revenue Uplift} \\
\midrule
Baseline ($\lambda=0$) & 2.38\% & 21.8\% & $1.0\times$ \\
GEM-Rec ($\lambda=0.5$) & 7.01\% & 81.5\% & $9.0\times$ \\
GEM-Rec ($\lambda=1.0$) & 18.03\% & 97.4\% & $28.2\times$ \\
GEM-Rec ($\lambda=2.0$) & 62.44\% & 99.9\% & $104.6\times$ \\
\bottomrule
\end{tabular}
}
\end{table}

\begin{table}[h]
\centering
\caption{\textbf{Bid Shock Response (Toys).} At $\lambda=0$, the High-Value Share is extremely low (3.2\%), indicating the model is effectively blind to the new bids. GEM-Rec corrects this rapidly, reaching 97\% share by $\lambda=2.0$.}
\label{tab:shock_toys}
\resizebox{0.6\textwidth}{!}{
\begin{tabular}{l ccc}
\toprule
\textbf{Setting} & \textbf{Ad Rate} & \textbf{High-Value Share} & \textbf{Revenue Uplift} \\
\midrule
Baseline ($\lambda=0.0$) & 3.49\% & 3.2\% & $1.0\times$ \\
GEM-Rec ($\lambda=0.5$) & 10.31\% & 49.3\% & $14.3\times$ \\
GEM-Rec ($\lambda=1.0$) & 24.45\% & 81.7\% & $55.8\times$ \\
GEM-Rec ($\lambda=2.0$) & 68.66\% & 97.2\% & $202.4\times$ \\
\bottomrule
\end{tabular}
}
\end{table}

\begin{table}[h]
\centering
\caption{\textbf{Bid Shock Response (Beauty).} The model demonstrates consistent adaptability. The High-Value share jumps from a negligible 1.6\% to 43\% with only a minor increase in Ad Rate (from 3.1\% to 8.8\%).}
\label{tab:shock_beauty}
\resizebox{0.6\textwidth}{!}{
\begin{tabular}{l ccc}
\toprule
\textbf{Setting} & \textbf{Ad Rate} & \textbf{High-Value Share} & \textbf{Revenue Uplift} \\
\midrule
Baseline ($\lambda=0.0$) & 3.12\% & 1.6\% & $1.0\times$ \\
GEM-Rec ($\lambda=0.5$) & 8.76\% & 43.0\% & $14.7\times$ \\
GEM-Rec ($\lambda=1.0$) & 20.51\% & 77.6\% & $57.1\times$ \\
GEM-Rec ($\lambda=2.0$) & 63.65\% & 98.5\% & $240.6\times$ \\
\bottomrule
\end{tabular}
}
\end{table}

\begin{table}[h]
\centering
\caption{\textbf{Bid Shock Response (Sports).} This dataset exhibits the most aggressive adaptation. At $\lambda=1.0$, nearly every ad shown (99.2\%) is already from the high-value subset.}
\label{tab:shock_sports}
\resizebox{0.6\textwidth}{!}{
\begin{tabular}{l ccc}
\toprule
\textbf{Setting} & \textbf{Ad Rate} & \textbf{High-Value Share} & \textbf{Revenue Uplift} \\
\midrule
Baseline ($\lambda=0.0$) & 4.87\% & 4.7\% & $1.0\times$ \\
GEM-Rec ($\lambda=0.5$) & 12.56\% & 72.9\% & $14.6\times$ \\
GEM-Rec ($\lambda=1.0$) & 30.41\% & 99.2\% & $50.1\times$ \\
GEM-Rec ($\lambda=2.0$) & 74.71\% & 100.0\% & $136.2\times$ \\
\bottomrule
\end{tabular}
}
\end{table}

\section{Ablations: Higher ad rate in training set}
\label{app:ablations}

In this section, we present the complete visual analysis for the \textbf{High Ad Rate Ablation}. We include all numerical result in Table \ref{tab:appendix_full_results_strict}.

This setting simulates a marketplace where the historical interaction logs already contain a significantly higher density of ad interactions (approximately 10--20\%, compared to 3--6\% in the main experiment). These results confirm that GEM-Rec successfully learns this higher baseline frequency directly from the data (evident at $\lambda=0$) and maintains the ability to perform precise trade-offs on top of this elevated base rate.

\subsubsection{Steam Dataset (High Rate)}
\begin{figure*}[h!]
    \centering
    \begin{subfigure}[b]{0.48\textwidth}
        \centering
        \includegraphics[width=0.95\linewidth]{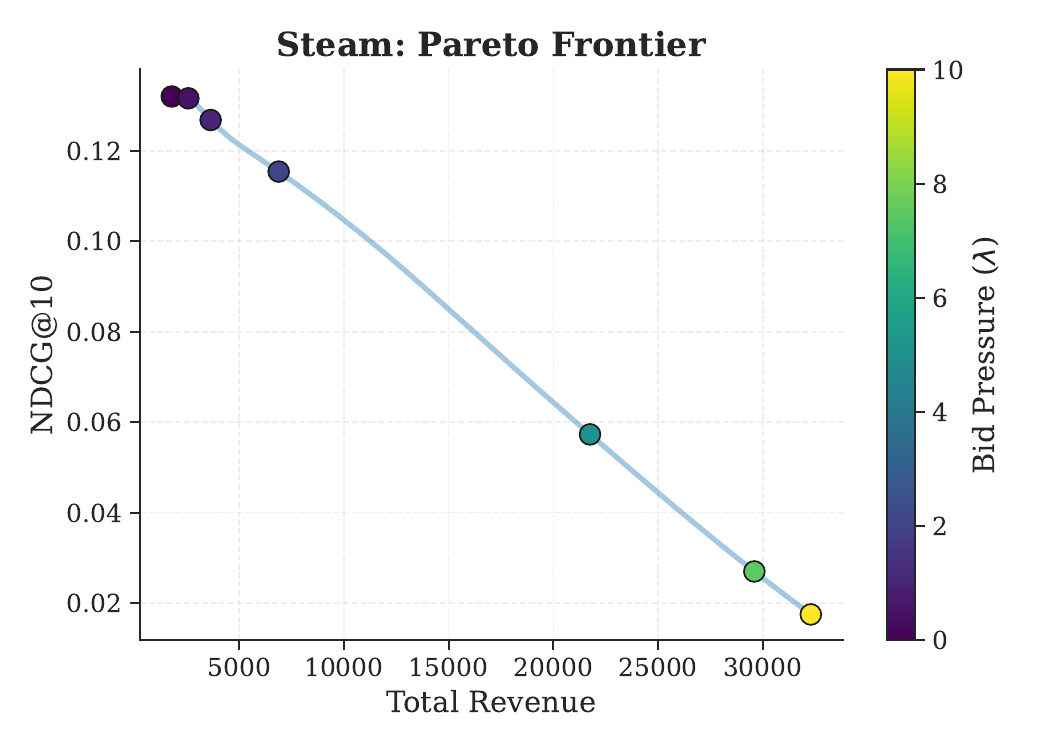}
        \caption{\textbf{Pareto Frontier} (Steam High-Rate)}
    \end{subfigure}
    \hfill
    \begin{subfigure}[b]{0.48\textwidth}
        \centering
        \includegraphics[width=0.95\linewidth]{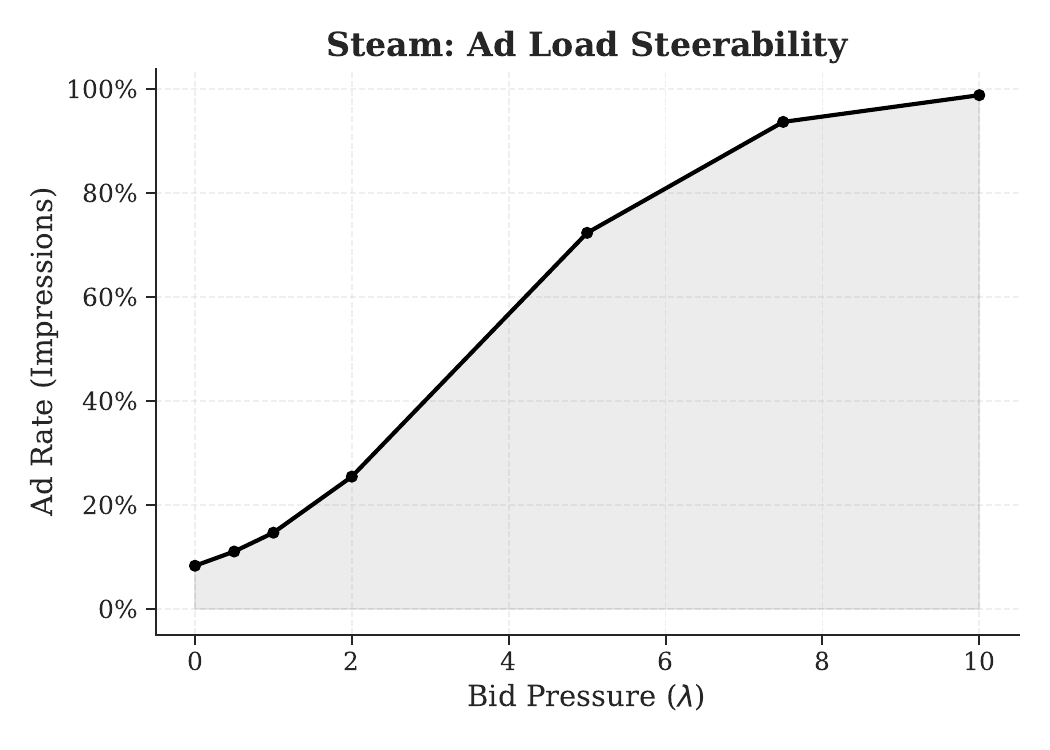}
        \caption{\textbf{Steerability} (Steam High-Rate)}
    \end{subfigure}
    
    \vspace{1em}
    
    \begin{subfigure}[b]{0.44\textwidth}
        \centering
        \includegraphics[width=0.95\linewidth]{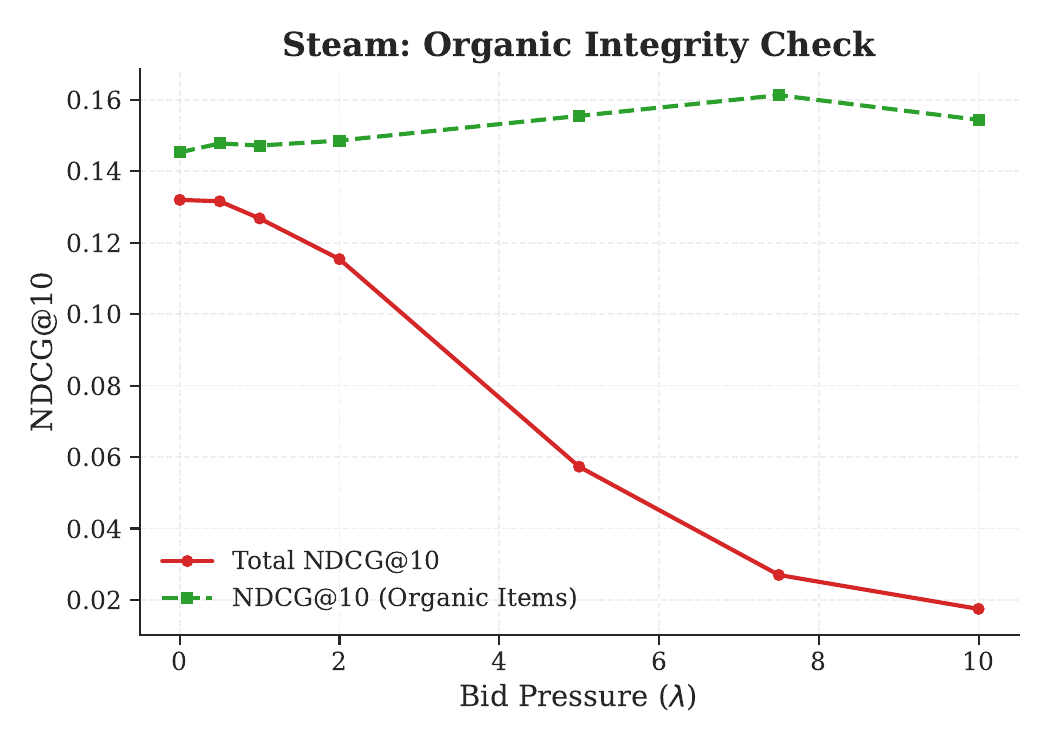}
        \caption{\textbf{Organic Integrity} (Steam High-Rate)}
    \end{subfigure}
    \hfill
    \begin{subfigure}[b]{0.5\textwidth}
        \centering
        \includegraphics[width=0.95\linewidth]{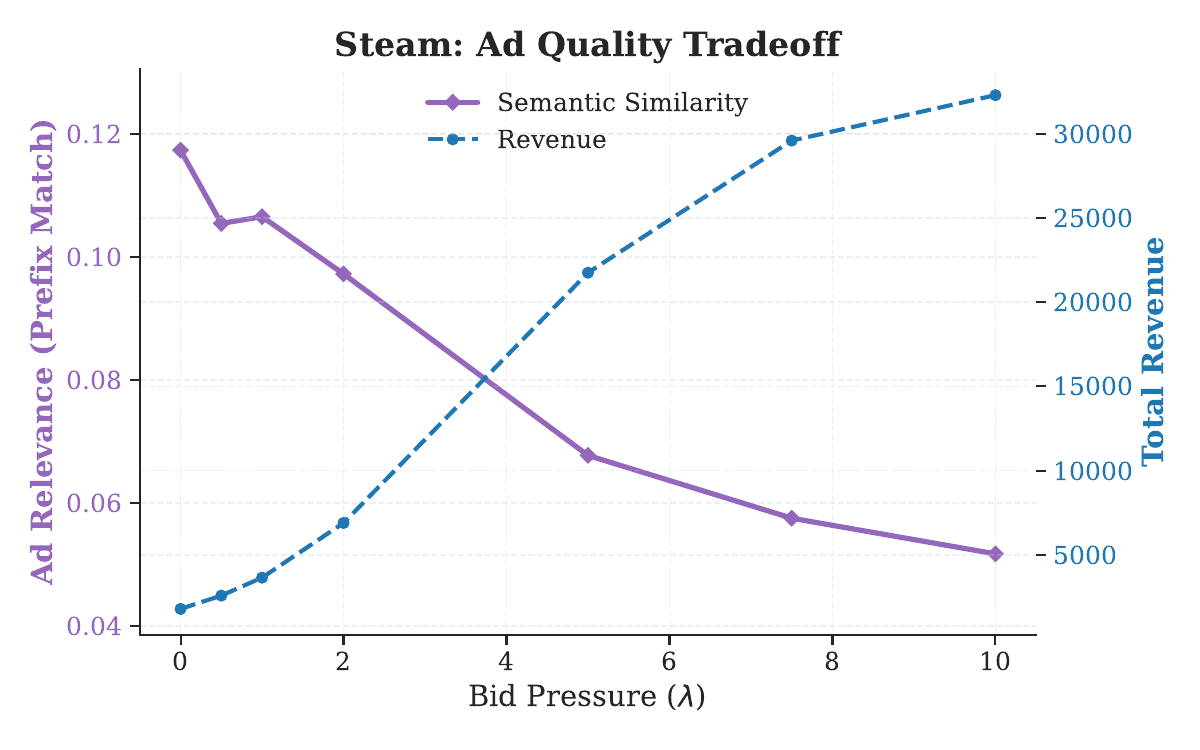}
        \caption{\textbf{Quality Trade-off} (Steam High-Rate)}
    \end{subfigure}
    
    \caption{\textbf{Dynamics on Steam (High Ad Rate).} The model correctly learns the higher base ad rate ($\sim$10\%) at $\lambda=0$ and allows strictly monotonic boosting beyond this point.}
    \label{fig:steam_high}
\end{figure*}
\clearpage

\subsubsection{Beauty Dataset (High Rate)}
\begin{figure*}[h!]
    \centering
    \begin{subfigure}[b]{0.48\textwidth}
        \centering
        \includegraphics[width=0.95\linewidth]{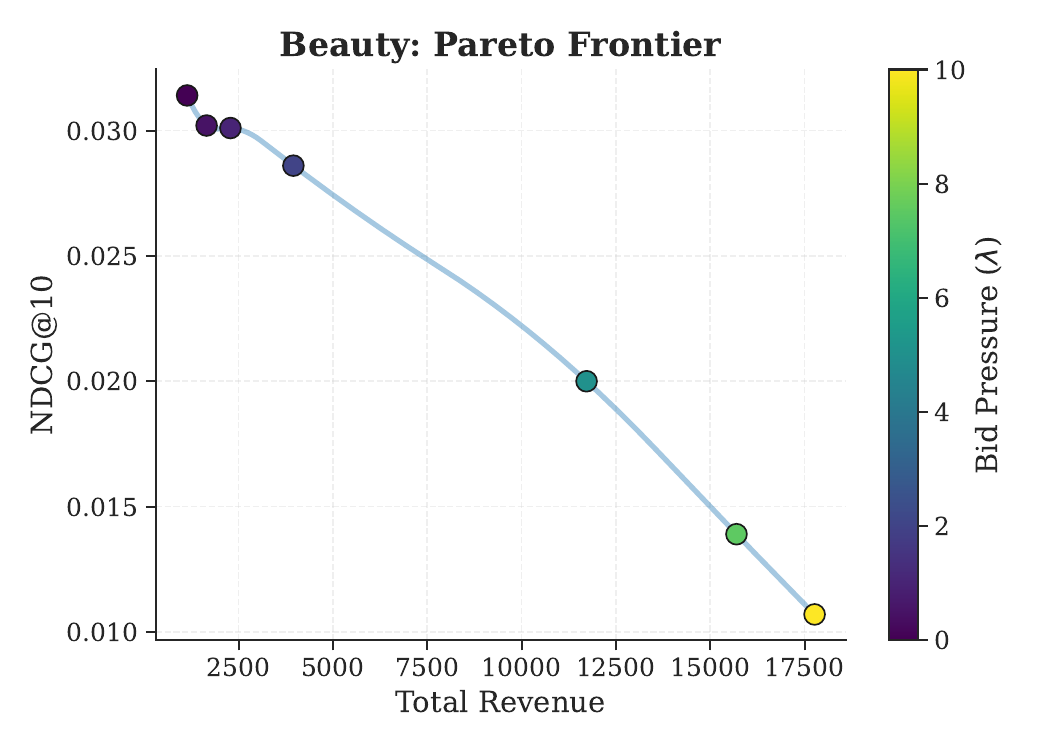}
        \caption{\textbf{Pareto Frontier} (Beauty High-Rate)}
    \end{subfigure}
    \hfill
    \begin{subfigure}[b]{0.48\textwidth}
        \centering
        \includegraphics[width=0.95\linewidth]{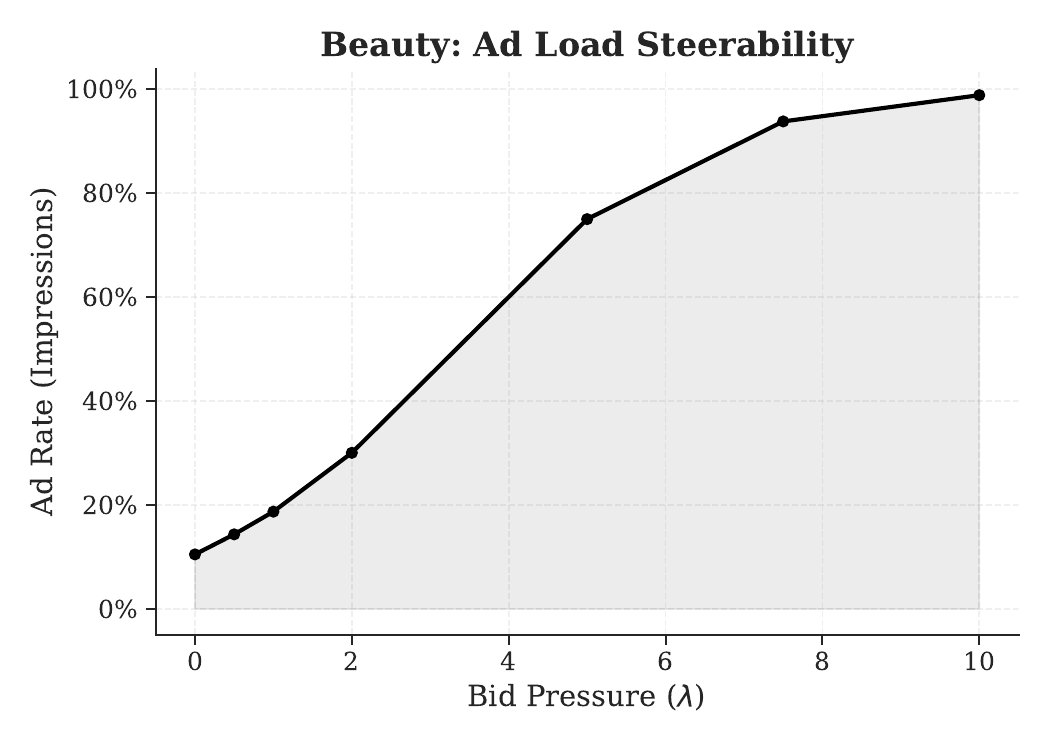}
        \caption{\textbf{Steerability} (Beauty High-Rate)}
    \end{subfigure}
    
    \vspace{1em}
    
    \begin{subfigure}[b]{0.44\textwidth}
        \centering
        \includegraphics[width=0.95\linewidth]{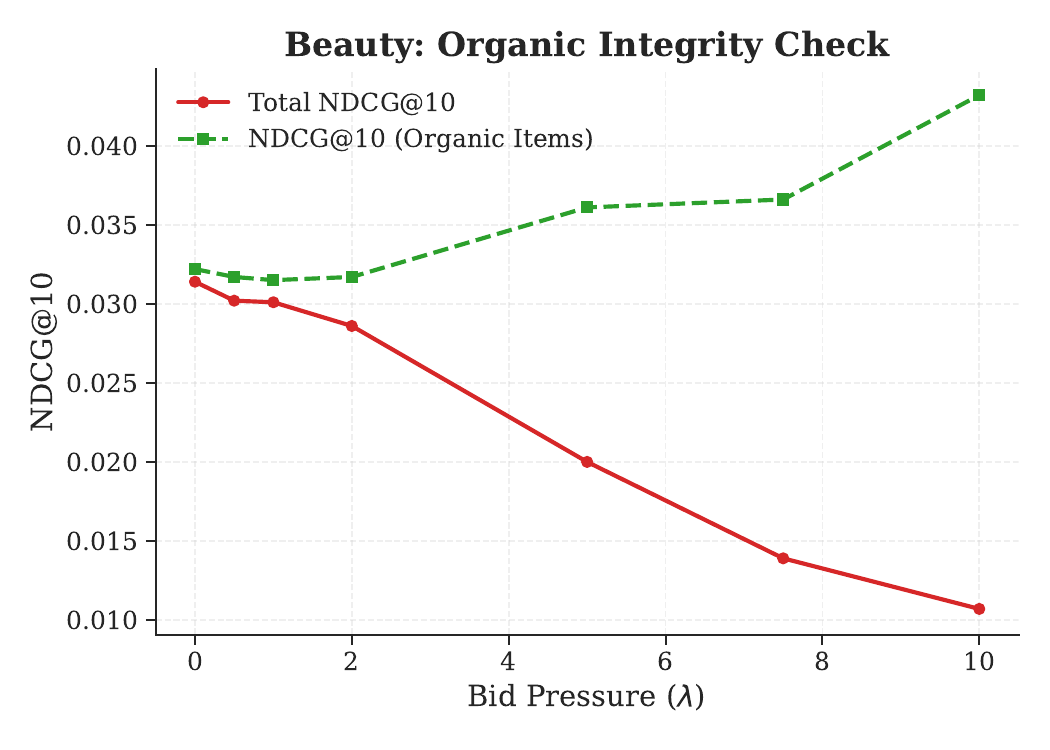}
        \caption{\textbf{Organic Integrity} (Beauty High-Rate)}
    \end{subfigure}
    \hfill
    \begin{subfigure}[b]{0.5\textwidth}
        \centering
        \includegraphics[width=0.95\linewidth]{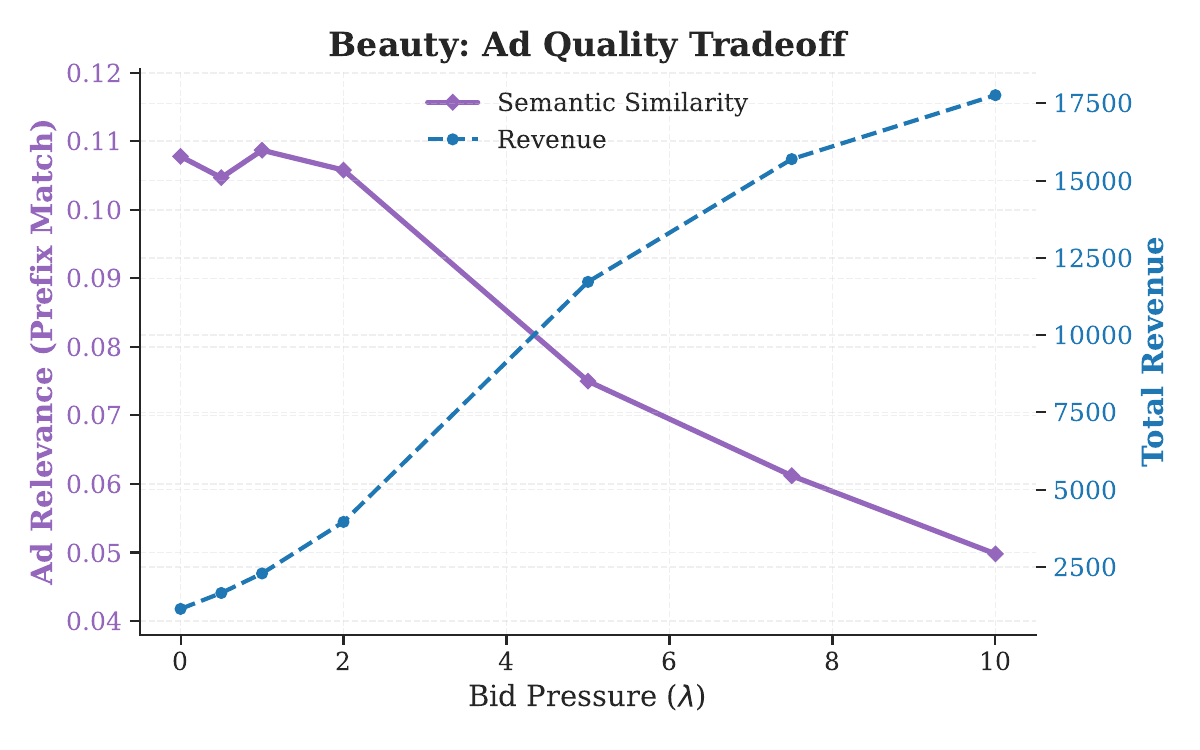}
        \caption{\textbf{Quality Trade-off} (Beauty High-Rate)}
    \end{subfigure}
    
    \caption{\textbf{Dynamics on Beauty (High Ad Rate).}}
    \label{fig:beauty_high}
\end{figure*}
\clearpage

\subsubsection{Sports Dataset (High Rate)}
\label{app:convex_eg}
\begin{figure*}[h!]
    \centering
    \begin{subfigure}[b]{0.48\textwidth}
        \centering
        \includegraphics[width=0.95\linewidth]{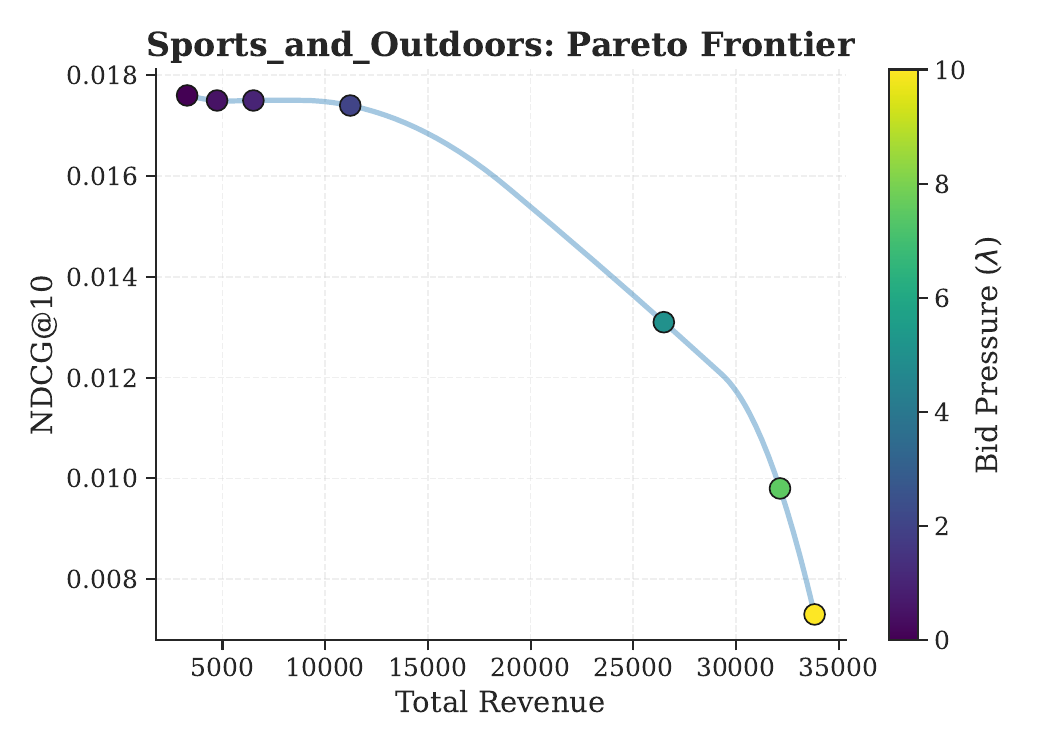}
        \caption{\textbf{Pareto Frontier} (Sports High-Rate)}
    \end{subfigure}
    \hfill
    \begin{subfigure}[b]{0.48\textwidth}
        \centering
        \includegraphics[width=0.95\linewidth]{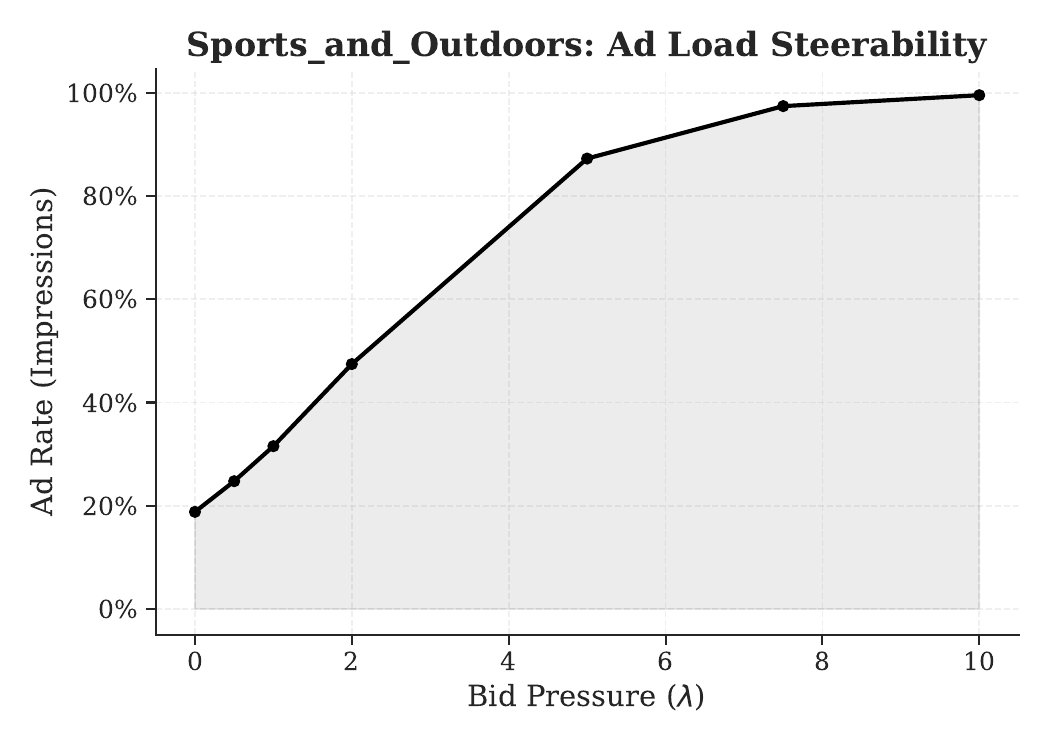}
        \caption{\textbf{Steerability} (Sports High-Rate)}
    \end{subfigure}
    
    \vspace{1em}
    
    \begin{subfigure}[b]{0.44\textwidth}
        \centering
        \includegraphics[width=0.95\linewidth]{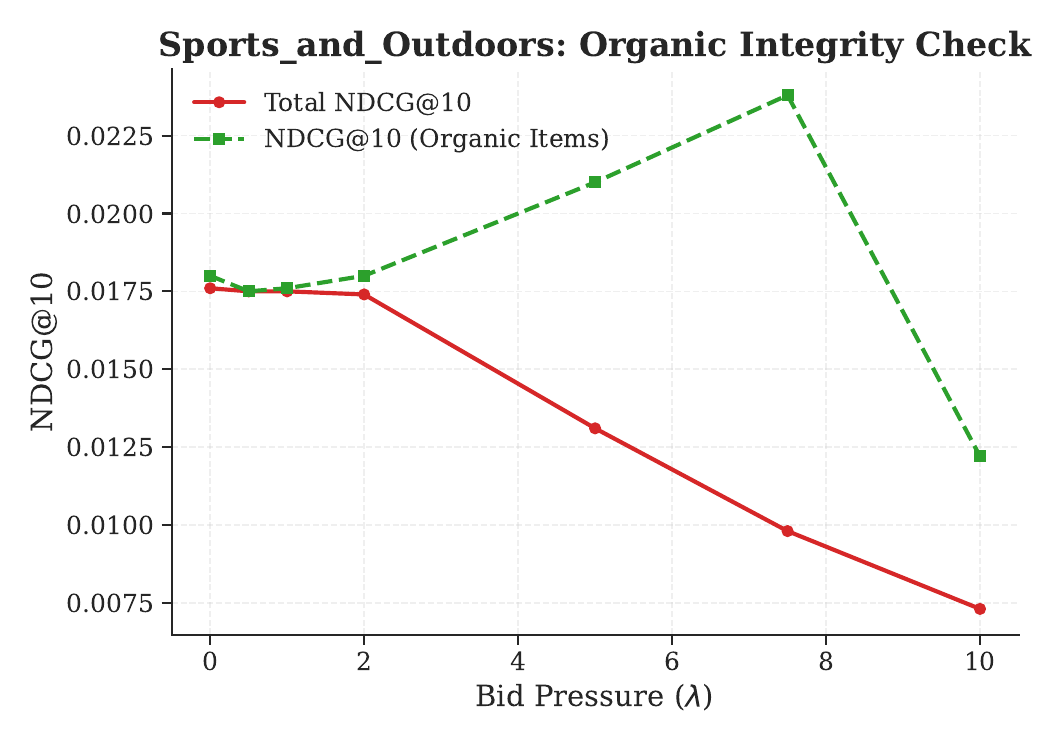}
        \caption{\textbf{Organic Integrity} (Sports High-Rate).}
    \end{subfigure}
    \hfill
    \begin{subfigure}[b]{0.5\textwidth}
        \centering
        \includegraphics[width=0.95\linewidth]{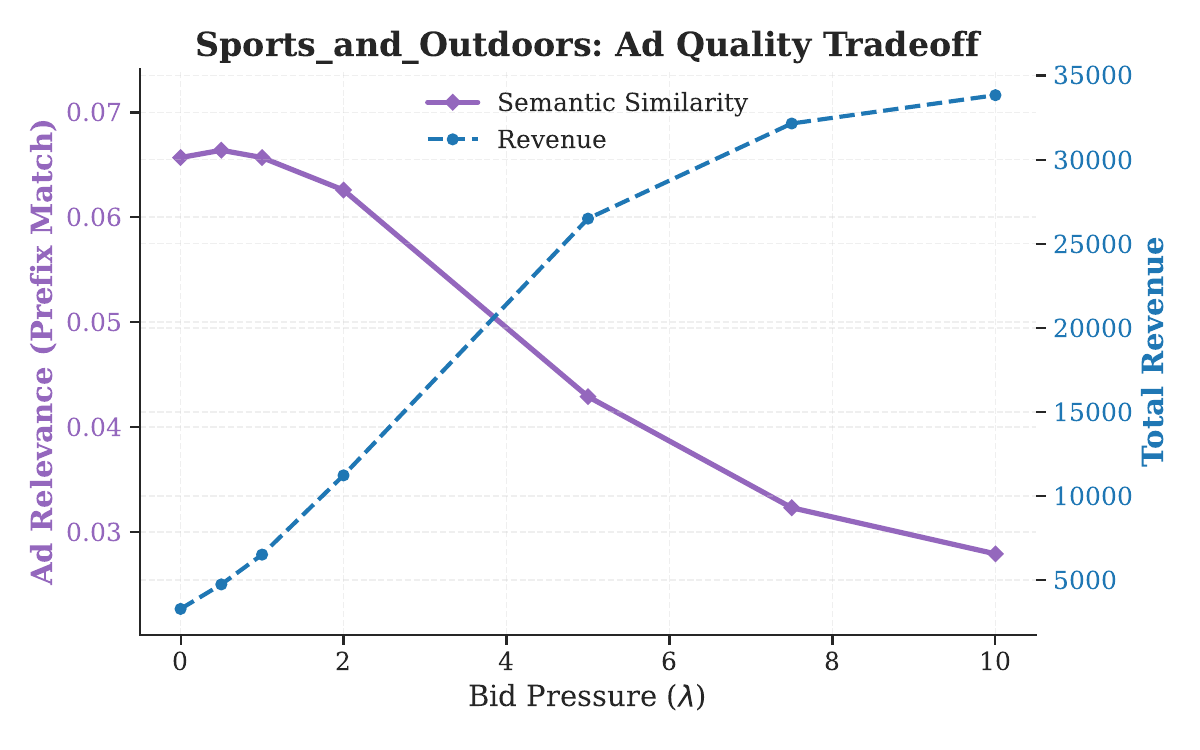}
        \caption{\textbf{Quality Trade-off} (Sports High-Rate)}
    \end{subfigure}
    
    \caption{\textbf{Dynamics on Sports (High Ad Rate). Note here that the drop of the green dotted curve at $\lambda=10$ for figure c) here corresponds with a nearly 100\% ad rate (see figure b) leaving very few organic generation, therefore it reflects a high variance and does not contradict our main conclusion and observations.}}
    \label{fig:sports_high}
\end{figure*}
\clearpage

\subsubsection{Toys Dataset (High Rate)}
\label{app:convex_eg2}
\begin{figure*}[h!]
    \centering
    \begin{subfigure}[b]{0.48\textwidth}
        \centering
        \includegraphics[width=0.95\linewidth]{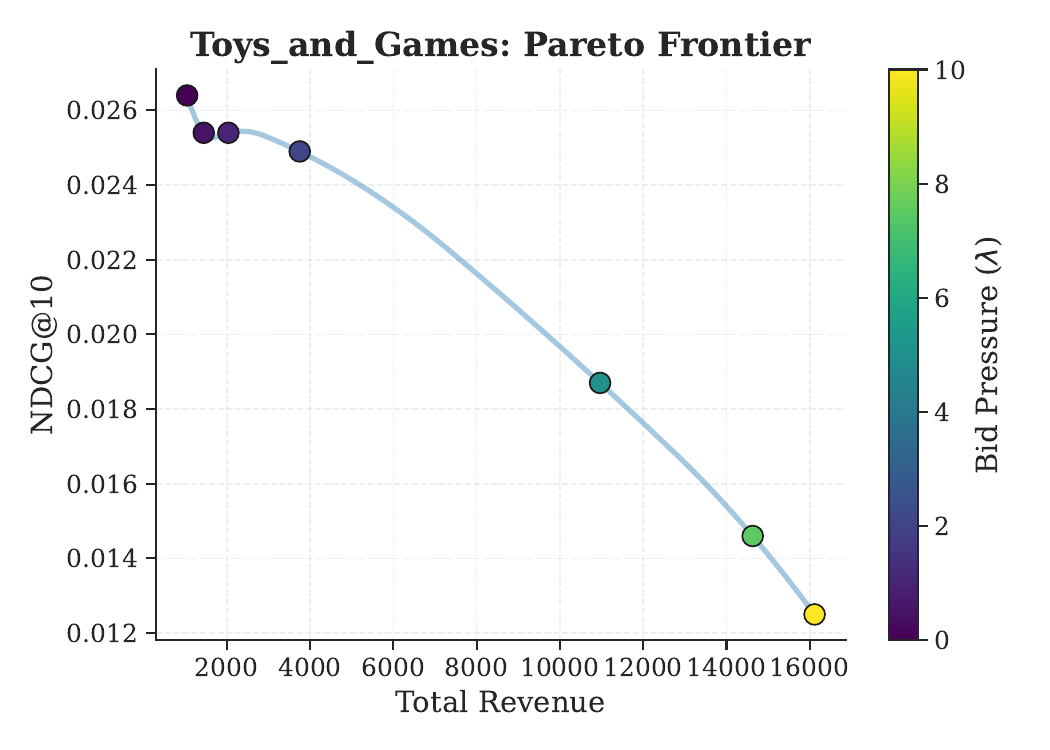}
        \caption{\textbf{Pareto Frontier} (Toys High-Rate)}
    \end{subfigure}
    \hfill
    \begin{subfigure}[b]{0.48\textwidth}
        \centering
        \includegraphics[width=0.95\linewidth]{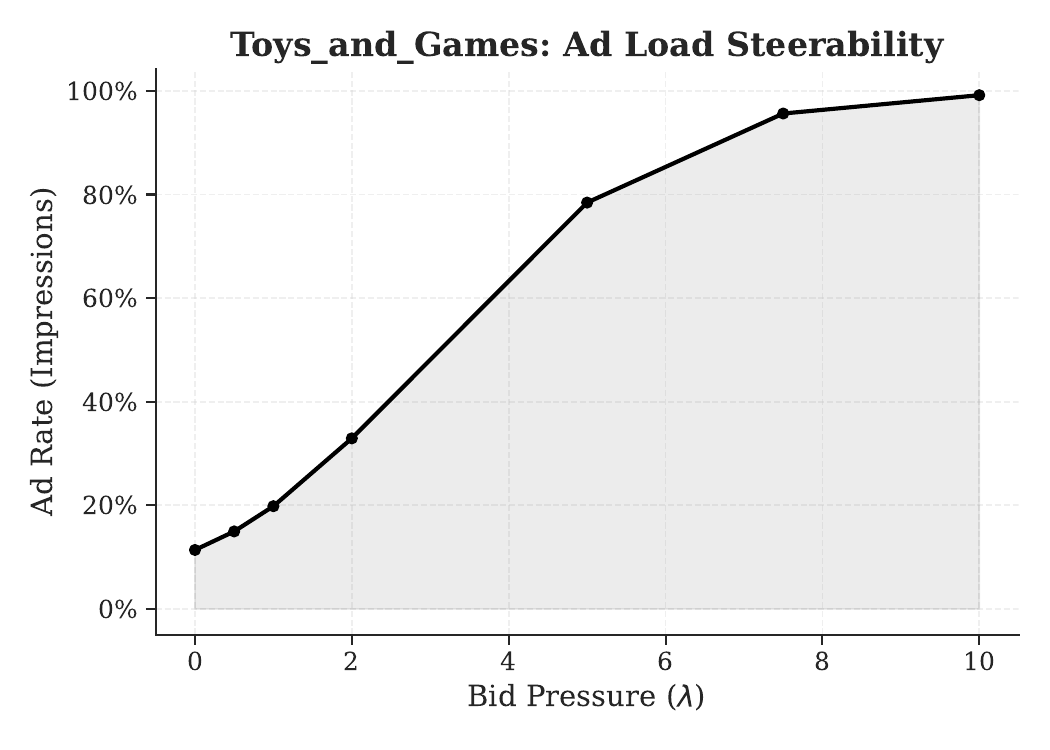}
        \caption{\textbf{Steerability} (Toys High-Rate)}
    \end{subfigure}
    
    \vspace{1em}
    
    \begin{subfigure}[b]{0.44\textwidth}
        \centering
        \includegraphics[width=0.95\linewidth]{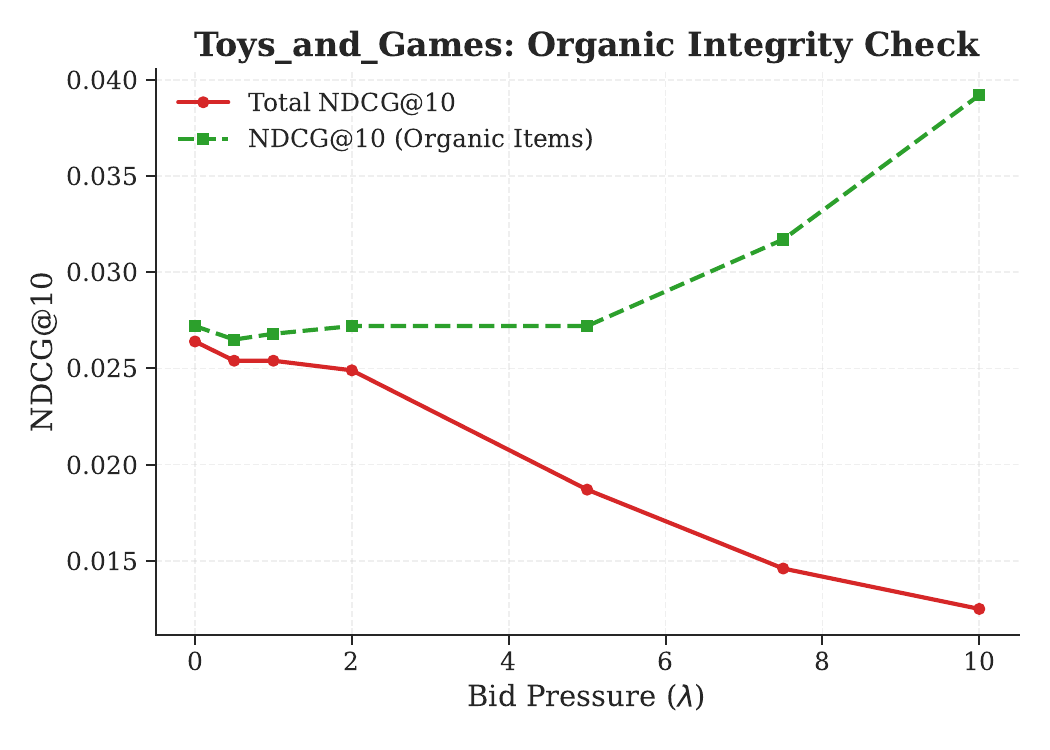}
        \caption{\textbf{Organic Integrity} (Toys High-Rate)}
    \end{subfigure}
    \hfill
    \begin{subfigure}[b]{0.5\textwidth}
        \centering
        \includegraphics[width=0.95\linewidth]{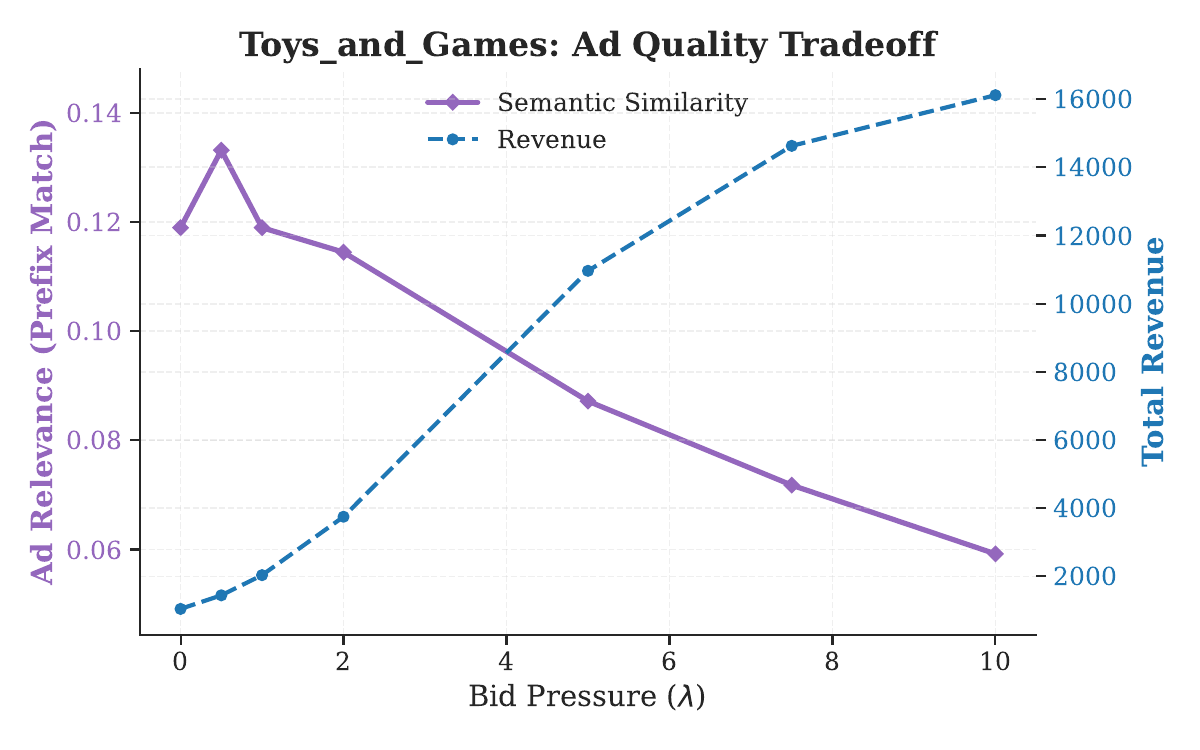}
        \caption{\textbf{Quality Trade-off} (Toys High-Rate)}
    \end{subfigure}
    
    \caption{\textbf{Dynamics on Toys (High Ad Rate).}}
    \label{fig:toys_high}
\end{figure*}
\clearpage

\begin{table*}[h]
\centering
\setlength{\tabcolsep}{4pt} 
\renewcommand{\arraystretch}{1.0} 
\caption{\textbf{Full Marketplace Performance Results (Strict Last-Item Protocol).} We compare GEM-Rec against the baseline (TIGER) across all datasets and $\lambda$ settings under the new evaluation setting.}
\label{tab:appendix_full_results_strict}
\resizebox{1.0\textwidth}{!}{
\begin{tabular}{ll c cc cc cc}
\toprule
& & \textbf{Train} & \multicolumn{2}{c}{\textbf{Economic Utility}} & \multicolumn{2}{c}{\textbf{Total Metrics}} & \multicolumn{2}{c}{\textbf{Metrics for Organic Gen.}} \\
\cmidrule(lr){4-5} \cmidrule(lr){6-7} \cmidrule(lr){8-9} 
\textbf{Dataset} & \textbf{Method} & \textbf{Ad \%} & \textbf{Ad Rate} & \textbf{Revenue} & \textbf{NDCG@10} & \textbf{Recall@10} & \textbf{O.NDCG@10} & \textbf{O.Recall@10} \\
\midrule

\multirow{8}{*}{\textbf{Beauty}} 
& TIGER (Baseline) & \multirow{8}{*}{\textbf{11.99\%}} & 0.0\% & 0.0 & 0.0245$^\dagger$ & 0.0458 & 0.0277 & 0.0518 \\
& GEM-Rec ($\lambda=0.0$) & & 10.5\% & 1,137.1 & 0.0314 & 0.0565 & 0.0322 & 0.0589 \\
& GEM-Rec ($\lambda=0.5$) & & 14.3\% & 1,653.0 & 0.0302 & 0.0546 & 0.0317 & 0.0581 \\
& GEM-Rec ($\lambda=1.0$) & & 18.7\% & 2,285.2 & 0.0301 & 0.0539 & 0.0315 & 0.0578 \\
& GEM-Rec ($\lambda=2.0$) & & 30.0\% & 3,953.5 & 0.0286 & 0.0517 & 0.0317 & 0.0587 \\
& GEM-Rec ($\lambda=5.0$) & & 75.0\% & 11,722.6 & 0.0200 & 0.0364 & 0.0361 & 0.0649 \\
& GEM-Rec ($\lambda=7.5$) & & 93.8\% & 15,695.7 & 0.0139 & 0.0271 & 0.0366 & 0.0665 \\
& GEM-Rec ($\lambda=10.0$) & & 98.9\% & 17,765.6 & 0.0107 & 0.0216 & 0.0432 & 0.0748 \\
\midrule

\multirow{8}{*}{\textbf{Sports}} 
& TIGER (Baseline) & \multirow{8}{*}{\textbf{20.73\%}} & 0.0\% & 0.0 & 0.0139$^\dagger$ & 0.0245 & 0.0174 & 0.0307 \\
& GEM-Rec ($\lambda=0.0$) & & 18.8\% & 3,279.4 & 0.0176 & 0.0330 & 0.0180 & 0.0340 \\
& GEM-Rec ($\lambda=0.5$) & & 24.8\% & 4,736.2 & 0.0175 & 0.0326 & 0.0175 & 0.0332 \\
& GEM-Rec ($\lambda=1.0$) & & 31.5\% & 6,506.6 & 0.0175 & 0.0330 & 0.0176 & 0.0334 \\
& GEM-Rec ($\lambda=2.0$) & & 47.4\% & 11,221.5 & 0.0174 & 0.0331 & 0.0180 & 0.0342 \\
& GEM-Rec ($\lambda=5.0$) & & 87.3\% & 26,486.6 & 0.0131 & 0.0280 & 0.0210 & 0.0389 \\
& GEM-Rec ($\lambda=7.5$) & & 97.4\% & 32,143.2 & 0.0098 & 0.0223 & 0.0238 & 0.0439 \\
& GEM-Rec ($\lambda=10.0$) & & 99.6\% & 33,827.0 & 0.0073 & 0.0172 & 0.0122 & 0.0190 \\
\midrule

\multirow{8}{*}{\textbf{Toys}} 
& TIGER (Baseline) & \multirow{8}{*}{\textbf{12.79\%}} & 0.0\% & 0.0 & 0.0235$^\dagger$ & 0.0451 & 0.0272 & 0.0520 \\
& GEM-Rec ($\lambda=0.0$) & & 11.4\% & 1,046.9 & 0.0264 & 0.0495 & 0.0272 & 0.0516 \\
& GEM-Rec ($\lambda=0.5$) & & 15.0\% & 1,446.0 & 0.0254 & 0.0468 & 0.0265 & 0.0500 \\
& GEM-Rec ($\lambda=1.0$) & & 19.8\% & 2,037.8 & 0.0254 & 0.0473 & 0.0268 & 0.0509 \\
& GEM-Rec ($\lambda=2.0$) & & 32.9\% & 3,751.0 & 0.0249 & 0.0447 & 0.0272 & 0.0503 \\
& GEM-Rec ($\lambda=5.0$) & & 78.5\% & 10,964.2 & 0.0187 & 0.0347 & 0.0272 & 0.0497 \\
& GEM-Rec ($\lambda=7.5$) & & 95.6\% & 14,630.0 & 0.0146 & 0.0294 & 0.0317 & 0.0626 \\
& GEM-Rec ($\lambda=10.0$) & & 99.2\% & 16,116.1 & 0.0125 & 0.0264 & 0.0392 & 0.0818 \\
\midrule
\multirow{8}{*}{\textbf{Steam}} 
& TIGER (Baseline) & \multirow{8}{*}{\textbf{10.13\%}} & 0.0\% & 0.0 & 0.1376$^\dagger$ & 0.1722 & 0.1515 & 0.1896 \\
& GEM-Rec ($\lambda=0.0$) & & 8.3\% & 1,782.7 & 0.1320 & 0.1661 & 0.1453 & 0.1818 \\
& GEM-Rec ($\lambda=0.5$) & & 11.0\% & 2,573.5 & 0.1316 & 0.1657 & 0.1478 & 0.1848 \\
& GEM-Rec ($\lambda=1.0$) & & 14.7\% & 3,638.6 & 0.1268 & 0.1597 & 0.1472 & 0.1841 \\
& GEM-Rec ($\lambda=2.0$) & & 25.4\% & 6,887.0 & 0.1154 & 0.1458 & 0.1486 & 0.1856 \\
& GEM-Rec ($\lambda=5.0$) & & 72.3\% & 21,748.5 & 0.0573 & 0.0768 & 0.1555 & 0.1903 \\
& GEM-Rec ($\lambda=7.5$) & & 93.7\% & 29,593.4 & 0.0270 & 0.0422 & 0.1613 & 0.1987 \\
& GEM-Rec ($\lambda=10.0$) & & 98.8\% & 32,289.6 & 0.0175 & 0.0307 & 0.1544 & 0.1860 \\

\bottomrule
\multicolumn{9}{l}{\footnotesize $^\dagger$ Imputed Score: Baseline strictly predicts Organic, effectively scoring 0 on all Ad slots.}
\end{tabular}
}
\end{table*}

\end{document}